\documentclass[12pt,a4paper]{article}
\usepackage{amssymb}
\usepackage{amsfonts}
\usepackage[onehalfspacing]{setspace}
\usepackage{pstricks}

\newtheorem{theorem}{Theorem}

\newtheorem{corollary}[theorem]{Corollary}

\newtheorem{lemma}[theorem]{Lemma}

\newtheorem{proposition}[theorem]{Proposition}

\newenvironment{proof}[1][Proof]{\noindent\textbf{#1.} }{\hfill $\square$}
\input{tcilatex}
\begin{document}

\title{Smooth Calibration, Leaky Forecasts, Finite Recall, and Nash Dynamics%
\thanks{%
Previous versions: July 2012, February 2015, March 2017. Research of the
second author was partially supported by a European Research Council (ERC)
Advanced Investigator grant. The authors thank Yakov Babichenko for useful
comments, and the editor, associate editor, and referees for their very
careful reading and helpful suggestions.}}
\author{Dean P. Foster\thanks{%
Amazon Inc, New York City, and University of Pennsylvania, Philadelphia. 
\emph{e-mail}: \texttt{dean@foster.net} \ \emph{web page}: \texttt{%
http://deanfoster.net/}} \and Sergiu Hart\thanks{%
Institute of Mathematics, Department of Economics, and Center for the Study
of Rationality, The Hebrew University of Jerusalem. \emph{e-mail}: \texttt{%
hart@huji.ac.il} \ \emph{web page}: \texttt{http://www.ma.huji.ac.il/hart}}}
\date{January 12, 2018}
\maketitle

\begin{abstract}
We propose to smooth out the calibration score, which measures how good a
forecaster is, by combining nearby forecasts. While regular calibration can
be guaranteed only by randomized forecasting procedures, we show that \emph{%
smooth calibration} can be guaranteed by \emph{deterministic} procedures. As
a consequence, it does not matter if the forecasts are \emph{leaked}, i.e.,
made known in advance: smooth calibration can nevertheless be guaranteed
(while regular calibration cannot). Moreover, our procedure has finite
recall, is stationary, and all forecasts lie on a finite grid. To construct
the procedure, we deal also with the related setups of online linear
regression and weak calibration. Finally, we show that smooth calibration
yields uncoupled finite-memory \emph{dynamics} in $n$-person
games---\textquotedblleft smooth calibrated learning"---in which the players
play approximate \emph{Nash equilibria} in almost all periods (by contrast,
calibrated learning, which uses regular calibration, yields only that the
time-averages of play are approximate correlated equilibria).
\end{abstract}

\tableofcontents

\def\@biblabel#1{#1\hfill}
\def\thebibliography#1{\section*{References}
\addcontentsline{toc}{section}{References}
\list
{}{
\labelwidth 0pt
\leftmargin 1.8em
\itemindent -1.8em
\usecounter{enumi}}
\def\newblock{\hskip .11em plus .33em minus .07em}
\sloppy\clubpenalty4000\widowpenalty4000
\sfcode`\.=1000\relax\def\baselinestretch{1}\large \normalsize}
\let\endthebibliography=\endlist%
\newcommand{\shfrac}[2]{\ensuremath{{}^{#1} \hspace{-0.04in}/_{\hspace{-0.03in}#2}}}%
\newcommand\T{\rule{0pt}{2.6ex}}
\newcommand\B{\rule[-1.2ex]{0pt}{0pt}}%

\section{Introduction\label{s:introduction}}

How good is a forecaster? Assume for concreteness that every day the
forecaster issues a forecast of the type \textquotedblleft the chance of
rain tomorrow is $30\%.$" A simple test one may conduct is to calculate the
proportion of rainy days out of those days for which the forecast was $30\%,$
and compare it to $30\%;$ and do the same for all other forecasts. A
forecaster is said to be \emph{calibrated }if, in the long run, the
differences between the actual proportions of rainy days and the forecasts
are small---no matter what the weather really was (see Dawid 1982).

What if rain is replaced by an event that is under the control of another
agent? If the forecasts are made public before the agent decides on his
action---we refer to this setup as \emph{\textquotedblleft leaky forecasts"}%
---then calibration \emph{cannot} be guaranteed; for example, the agent can
make the event happen if and only if the forecast is less than $50\%$, and
so the forecasting error (that is, the \textquotedblleft calibration score")
is always at least $50\%.$ However, if in each period the forecast and the
agent's decision are made \textquotedblleft simultaneously"---which means
that neither one knows the other's decision before making his own---then
calibration \emph{can} be guaranteed; see Foster and Vohra (1998). The
procedure that yields calibration no matter what the agent's decisions are
requires the use of \emph{randomizations} (e.g., with probability $1/2$ the
forecaster announces $30\%,$ and with probability $1/2$ he announces $60\%).$
Indeed, as the discussion at the beginning of this paragraph suggests, one
cannot have a deterministic procedure that is calibrated (see Dawid 1985 and
Oakes 1985).

Now the standard calibration score is overly fastidious: the days when the
forecast was, say, $30.01\%$ are considered separately from the days when
the forecast was $29.99\%$ (formally, the calibration score is a highly
discontinuous function of the data, i.e., the forecasts and the actions).
This suggests that one first combines all days when the forecast was \emph{%
close to} $30\%,$ and only then compares the $30\%$ with the proportion of
rainy days. If, say, there were $200$ days with a forecast of $30.01\%,$ out
of which $10$ were rainy, and another $100$ days with a forecast of $29.99\%$%
, out of which $80$ were rainy, then the forecaster is very far from being
calibrated; however, he is smoothly calibrated, as his forecasts were all
close to $30\%,$ and there were $90/300=30\%$ rainy days. Undershooting at $%
29.99\%$ and overshooting at $30.01\%$ is now balanced out. Formally, what
this amounts to is applying a so-called \textquotedblleft smoothing"
operation to the forecasting errors (which makes smooth calibration easier
to obtain than calibration).\footnote{%
Corollary \ref{c:K-to-K-lambda} in Section \ref{s:weak calibration} will
formally show that regular calibration implies smooth calibration.}

Perhaps surprisingly, once we consider smooth calibration, there is no
longer a need for randomization when making the forecasts: we will show that
there exist \emph{deterministic} procedures that guarantee smooth
calibration, no matter what the agent does. In particular, it follows that
it does not matter if the forecasts are made known to the agent before his
decision, and so smooth calibration can be guaranteed even when forecasts
may be leaked.\footnote{%
When the forecasting procedure is deterministic it can be simulated by the
agent, and so it is irrelevant whether the agent \emph{observes} the
forecasts, or just \emph{computes} them by himself, before taking his action.%
} This may come as a surprise, because, as pointed out above, an agent who
knows the forecast \emph{before} deciding on the weather will choose rain
when the forecast is less than $50\%$ and no rain otherwise, giving a
calibration error of $50\%$ or more, no matter what the forecaster does.
However, against such an agent one can easily be \emph{smoothly} calibrated,
by forecasting $50.01\%$ on odd days and $49.99\%$ on even days (the
resulting weather will then alternate between rain and no rain, and so half
the days will be rainy days---and all the forecasts are indeed close to $%
50\% $). What this proves is only that one can be smoothly calibrated
against this specific strategy of the agent (this is the strategy that shows
that it is impossible to have calibration with deterministic leaky
procedures); our result shows that one can in fact \emph{guarantee} smooth
calibration with a deterministic strategy, against \emph{any} strategy of
the agent.

The forecasting procedure that we construct and that guarantees smooth
calibration has moreover finite recall (i.e., only the forecasts and actions
of the last $R$ periods are taken into account, for some fixed finite $R$),
and is stationary (i.e., independent of \textquotedblleft calendar time":
the forecast is the same any time that the \textquotedblleft window" of the
past $R\ $periods is the same).\footnote{%
Another, seemingly less elegant, way to obtain this is by restarting the
procedure once in a while; see, e.g., Lehrer and Solan (2009).} Finally, we
can have all the forecasts lie on some finite fixed grid.

The construction starts with the \textquotedblleft online linear regression"
problem, introduced by Foster (1991), where one wants to generate every
period a good linear estimator based only on the data up to that point. We
provide a finite-recall stationary algorithm for this problem; see Section %
\ref{s:linear regression}. We then use this algorithm, together with a
fixed-point argument, to obtain \textquotedblleft weak calibration," a
concept introduced by Kakade and Foster (2004) and Foster and Kakade (2006);
see Section \ref{s:weak calibration}. Section \ref{s:smooth calibration}
shows that weak and smooth calibration are essentially equivalent, which
yields the existence of smoothly calibrated procedures. Finally, these
procedures are used to obtain dynamics (\textquotedblleft smoothly
calibrated learning") that are uncoupled, have finite memory, and are close
to Nash equilibria most of the time (while the similar dynamics that are
based on regular calibration yield only the time average becoming close to
correlated equilibria; see Foster and Vohra 1997).

\subsection{Literature\label{sus:literature}}

The \emph{calibration problem} has been extensively studied, starting with
Dawid (1982), Oakes (1985), and Foster and Vohra (1998); see Olszewski
(2015) for a comprehensive survey of the literature. Kakade and Foster
(2004) and Foster and Kakade (2006) introduced the notion of \emph{weak
calibration}, which shares many properties with smooth calibration. In
particular, both can be guaranteed by deterministic procedures, and both are
of the \textquotedblleft general fixed point" variety: they can find fixed
points of arbitrary continuous functions (see for instance the last
paragraph in Section \ref{sus:leaky}).\footnote{\label%
{ftn:fixed-pt-calibration}They are thus more \textquotedblleft powerful"
than the standard calibration procedures (such as those based on Blackwell's
approachability), which find \emph{linear }fixed points (such as
eigenvectors and invariant probabilities).} However, while weak calibration
may be at times technically more convenient to work with, smooth calibration
is the more natural concept, easier to interpret and understand; it is,
after all, just a standard smoothing of regular calibration.

The \emph{online regression problem}---see Section \ref{s:linear regression}
for details---was introduced by Foster (1991); for further improvements, see
J. Foster (1999), Vovk (2001), Azoury and Warmuth (2001), and the book of
Cesa-Bianchi and Lugosi (2006).

\section{Calibration: Model and Result\label{s:model}}

In this section we present the calibration game in its standard and
\textquotedblleft leaky" versions, introduce the notion of smooth
calibration, and state our main results.

\subsection{The Calibration Game\label{sus:calibration game}}

Let\footnote{%
We denote by $\mathbb{R}^{m}$ the $m$-dimensional Euclidean space, with the
usual $\ell _{2}$-norm $||\cdot ||$.} $C\subseteq \mathbb{R}^{m}$ be a
compact convex set, and let $A\subseteq C$ (for example, $C$ could be the
set of probability distributions $\Delta (A)$ over a finite set $A,$ which
is identified with the set of unit vectors in $C,$ or a product of such
sets). The \emph{calibration game} has two players: the \textquotedblleft
action" player---the \textquotedblleft A-player" for short---and the
\textquotedblleft conjecture" (or \textquotedblleft calibrating")
player---the \textquotedblleft C-player" for short. At each time period $%
t=1,2,...,$ the C-player chooses $c_{t}\in C$ and the A-player chooses $%
a_{t}\in A.$ There is full monitoring and perfect recall: at time $t$ both
players know the realized history $h_{t-1}=(c_{1},a_{1},...,c_{t-1},a_{t-1})%
\in (C\times A)^{t-1}.$

In the \emph{standard} calibration game, $c_{t}$ and $a_{t}$ are chosen
simultaneously (perhaps in a mixed, i.e., randomized, way). In the \emph{%
leaky} calibration game, $a_{t}$ is chosen after $c_{t}$ has been chosen and
revealed; thus, $c_{t}$ is a function of $h_{t-1},$ whereas $a_{t}$ is a
function of $h_{t-1}$ \emph{and }$c_{t}.$ Formally, a pure strategy of the
C-player is $\sigma :\cup _{t\geq 1}(C\times A)^{t-1}\rightarrow C,$ and a
pure strategy of the A-player is $\tau :\cup _{t\geq 1}(C\times
A)^{t-1}\rightarrow A$ in the standard game, and $\tau :\cup _{t\geq
1}(C\times A)^{t-1}\times C\rightarrow A$ in the leaky game. A pure strategy
of the C-player will also be referred to as \emph{deterministic.}

The calibration score---which the C-player wants to minimize---is defined%
\emph{\ }at time $T\geq 1$ as follows. For every forecast $c$ in $C$ let $%
n(c)\equiv n_{T}(c):=|\{1\leq t\leq T:c_{t}=c\}$ be the number of times that
it has been used, and let 
\[
\bar{a}(c)\equiv \bar{a}_{T}(c):=\frac{1}{n(c)}\sum_{s=1}^{T}\mathbf{1}%
_{c_{s}=c}\,a_{s} 
\]%
be the average of the actions in the periods when the forecast was $c,$
where we write $\mathbf{1}_{x=y}$ for the indicator that $x=y$ (i.e., $%
\mathbf{1}_{x=y}=1$ when $x=y$ and $\mathbf{1}_{x=y}=0$ otherwise); $\bar{a}%
(c)$ is defined only when $c$ appears in the sequence $c_{1},...,c_{T},$
i.e., $n(c)>0$. The \emph{calibration score} at time $T$ is then defined as%
\footnote{%
The summation is over all $c$ with $n(c)>0,$ and we use the Euclidean norm
(the squared distance $||\bar{a}(c)-c||^{2}$ may well be used instead, in
line with standard statistics usage).}%
\[
K_{T}:=\sum_{c\in C}\frac{n(c)}{T}||\bar{a}(c)-c||. 
\]%
Thus $K_{T}$ is the mean distance between the forecast $c$ and the average $%
\bar{a}(c)$ of the actions $a$ in those periods where the forecast was $c,$
weighted proportionately to how often each forecast $c$ has been used in
those $T$ periods. An alternative formulation is easily seen to be 
\begin{equation}
K_{T}=\frac{1}{T}\sum_{t=1}^{T}||\bar{a}_{t}-c_{t}||,  \label{eq:KT}
\end{equation}%
where $\bar{a}_{t}:=\bar{a}(c_{t}),$ i.e.,%
\[
\bar{a}_{t}:=\frac{\sum_{s=1}^{T}\mathbf{1}_{c_{s}=c_{t}}\,a_{s}}{%
\sum_{s=1}^{T}\mathbf{1}_{c_{s}=c_{t}}}; 
\]%
indeed, for each $c$ there are $n(c)$ identical terms in (\ref{eq:KT}) that
each equal $||\bar{a}(c)-c||.$

\subsection{Smooth Calibration\label{sus:smooth calibration}}

We introduce the notion of \textquotedblleft smooth calibration." A \emph{%
smoothing function} is a function $\Lambda :C\times C\rightarrow \mathbb{[}%
0,1]$ with $\Lambda (c,c)=1$ for every $c.$ Its interpretation is that $%
\Lambda (c^{\prime },c)$ gives the weight that we assign to $c^{\prime }$
when we are at $c.$ We will use $\Lambda (c^{\prime },c)$ instead of the
indicator $\mathbf{1}_{c^{\prime }=c}$ to \textquotedblleft smooth" out the
forecasts and the average actions. Specifically, put%
\[
\bar{a}_{t}^{\Lambda }:=\frac{\sum_{s=1}^{T}\Lambda (c_{s},c_{t})\,a_{s}}{%
\sum_{s=1}^{T}\Lambda (c_{s},c_{t})}\text{ \ \ and\ \ \ }c_{t}^{\Lambda }:=%
\frac{\sum_{s=1}^{T}\Lambda (c_{s},c_{t})\,c_{s}}{\sum_{s=1}^{T}\Lambda
(c_{s},c_{t})}. 
\]%
The $\Lambda $\emph{-smoothed calibration score} at time $T$ is then defined
as

\begin{equation}
K_{T}^{\Lambda }=\frac{1}{T}\sum_{t=1}^{T}||\bar{a}_{t}^{\Lambda
}-c_{t}^{\Lambda }||.  \label{eq:K-def}
\end{equation}

A standard (and useful) assumption is a Lipschitz condition: there exists $%
L<\infty $ such that \TEXTsymbol{\vert}$\Lambda (c^{\prime },c)-\Lambda
(c^{\prime \prime },c)|\leq L||c^{\prime }-c^{\prime \prime }||$ for all $%
c,c^{\prime },c^{\prime \prime }\in C.$ Thus, the functions $\Lambda (\cdot
,c)$ are uniformly Lipschitz: $\mathcal{L}(\Lambda (\cdot ,c))\leq L$ for
every $c\in C,$ where $\mathcal{L}(f):=\sup \{\left\Vert
f(x)-f(y)\right\Vert /\left\Vert x-y\right\Vert :x,y\in X,~x\neq y\}$
denotes the \emph{Lipschitz constant }of the function $f$ (if $f$ is not a
Lipschitz function then $\mathcal{L}(f)=+\infty ;$ when $\mathcal{L}(f)\leq
L $ we say that $f$ is $L$\emph{-Lipschitz}).

Two classic examples of Lipschitz smoothing functions are: (i) the so-called 
\emph{tent function }$\Lambda (c^{\prime },c)=[1-||c^{\prime }-c||/\delta
]_{+}$ for $\delta >0$, where $[z]_{+}:=\max \{z,0\}$; thus, only points $%
c^{\prime }$ within distance $\delta $ of $c$ are considered, and their
weight is proportional to the distance from $c$ 
\begin{figure}[htbp] \centering%
\begin{pspicture}(-0.05,0)(12,4)
 \psset{unit=0.5cm}
 \psline[linewidth=1pt]{->}(0,1)(10,1)
 \psline[linewidth=1pt]{->}(0,1)(0,7)
 \psline[linewidth=3pt,linecolor=red]{-c}(0,1)(3.8,1)
 \psline[linewidth=3pt,linecolor=red]{c-}(4.2,1)(9.5,1)
 \pscircle*[linecolor=red](4,5.6){0.2}
 \psline[linewidth=1pt]{}(4,0.8)(4,1.2)
 \rput(4,0.5){$c$}
 \rput(10.5,1){$c'$}
 \rput(-0.5,5.6){$1$}
 \psline[linewidth=1pt]{}(-0.1,5.6)(0.1,5.6)
 \psset{origin={14,0}}
 \psline[linewidth=1pt]{->}(0,1)(10,1)
 \psline[linewidth=1pt]{->}(0,1)(0,7)
 \psline[linewidth=1pt]{}(4,0.8)(4,1.2)
 \rput(18,0.5){$c$}
 \rput(24.5,1){$c'$}
 \rput(13.5,5.6){$1$}
 \psline[linewidth=1pt]{}(-0.1,5.6)(0.1,5.6)
 \psline[linewidth=3pt,linecolor=blue]{}(0,1)(2,1)(4,5.6)(6,1)(9.5,1)
 \psline[linewidth=1pt]{}(2,0.8)(2,1.2)
 \psline[linewidth=1pt]{}(6,0.8)(6,1.2)
 \rput(16,0.5){$c-\delta$}
 \rput(20,0.5){$c+\delta$}
\end{pspicture}%
\caption{\emph{Left}: The indicator function
$\textbf{1}_{x=c}$. 
\emph{Right}: The 
$\delta$-tent smoothing function
$\Lambda(x,c)=[1-||x-c||/\delta]_+$ with Lipschitz bound
$L=1/\delta$.\label{fg:Lambda}}%
\end{figure}
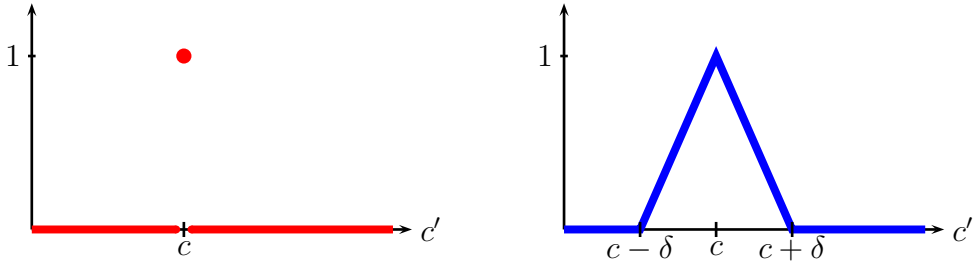%
(see Figure \ref{fg:Lambda} \emph{Right} for this function $\Lambda ,$ and
compare it with the indicator function in Figure \ref{fg:Lambda} \emph{Left}%
); and (ii) the so-called \emph{Guassian} (or \emph{normal}) smoothing
function $\Lambda (c^{\prime },c)=\exp (-||c^{\prime }-c||^{2}/(2\sigma
^{2}))$.

\bigskip

\noindent \textbf{Remarks. }\emph{(a)} The original calibration score $K_{T}$
is obtained when $\Lambda $ is the indicator function, i.e., $\Lambda
(c^{\prime },c)=\mathbf{1}_{c^{\prime }=c}$ for all $c,c^{\prime }\in C.$

\emph{(b)} The normalization $\Lambda (c,c)=1$ pins down the Lipschitz
constant (otherwise one could replace $\Lambda $ with $\alpha \Lambda $ for
small $\alpha >0,$ and so lower the Lipschitz constant without affecting the
score).

\emph{(c)} Smoothing both $\bar{a}_{t}$ and $c_{t}$ and then taking the
difference is the same as smoothing the difference: $\bar{a}_{t}^{\Lambda
}-c_{t}^{\Lambda }=(\bar{a}_{t}-c_{t})^{\Lambda }.$ Moreover, smoothing $%
a_{t}$ is the same as smoothing $\bar{a}_{t},$ i.e., $\bar{a}_{t}^{\Lambda
}=a_{t}^{\Lambda }.$

\emph{(d)} An alternative score smoothes only the average action $\bar{a}%
_{t} $, but not the forecast $c_{t}$:%
\[
\tilde{K}_{T}^{\Lambda }=\frac{1}{T}\sum_{t=1}^{T}||\bar{a}_{t}^{\Lambda
}-c_{t}||. 
\]%
If the smoothing function puts positive weight only in small neighborhoods,
i.e., there is $\delta >0$ such that $\Lambda (c^{\prime },c)>0$ only when $%
||c^{\prime }-c||\leq \delta ,$ then the difference between $K_{T}^{\Lambda
} $ and $\tilde{K}_{T}^{\Lambda }$ is at most $\delta $ (because in this
case $||c_{t}^{\Lambda }-c_{t}||\leq \delta $ for every $t).$ More
generally, $|K_{T}^{\Lambda }-\tilde{K}_{T}^{\Lambda }|\leq \delta $ when $%
(1/T)\sum_{t=1}^{T}||c_{t}^{\Lambda }-c_{t}||\leq \delta $ for any
collection of points $c_{1},...,c_{T}\in C,$ which is indeed the case, for
instance, for the Gaussian smoothing with small enough $\sigma ^{2}.$ The
reason that we prefer to use $K^{\Lambda }$ rather than $\tilde{K}^{\Lambda
} $ is that $K^{\Lambda }$ vanishes when there is perfect calibration (i.e., 
$\bar{a}_{t}=c_{t}$ for all $t),$ whereas $\tilde{K}^{\Lambda }$ remains
positive; clean statements such as $K_{T}^{\Lambda }\leq \varepsilon $
become $\tilde{K}_{T}^{\Lambda }\leq \varepsilon +\delta .$

\bigskip

Finally, given $\varepsilon >0$ and $L<\infty $, we will say that a strategy
of the C-player---which is also called a \textquotedblleft procedure"---is $%
(\varepsilon ,L)$\emph{-smoothly calibrated} if there is $T_{0}\equiv
T_{0}(\varepsilon ,L)$ such that%
\begin{equation}
K_{T}^{\Lambda }=\frac{1}{T}\sum_{t=1}^{T}\left\Vert \bar{a}_{t}^{\Lambda
}-c_{t}^{\Lambda }\right\Vert \leq \varepsilon  \label{eq:smooth-calib}
\end{equation}%
holds almost surely, for every strategy of the A-player, every $T>T_{0},$
and every smoothing function $\Lambda :C\times C\rightarrow \lbrack 0,1]$
that is $L$-Lipschitz in the first coordinate$.$ Unlike standard
calibration, which can be guaranteed only with high probability, smooth
calibration may be obtained by deterministic procedures---as will be shown
below---in which case we may well require (\ref{eq:smooth-calib}) to always
hold (rather than just almost surely).\textbf{\ }

\subsection{Leaky Forecasts\label{sus:leaky}}

We will say that a procedure (i.e., a strategy of the C-player) is \emph{%
leaky }(smoothly) calibrated if it is (smoothly) calibrated also in the
leaky setup, that is, against an A-player who may choose his action $a_{t}$
at time $t$ depending on the forecast $c_{t}$ made by the C-player at time $%
t $ (i.e., the A-player moves after the C-player). While, as we saw in the
Introduction, there are no leaky calibrated procedures, we will show that
there are \emph{leaky smoothly calibrated} procedures.

Deterministic procedures (i.e., pure strategies of the C-player) are clearly
leaky: the A-player can use the procedure at each period $t$ to compute $%
c_{t}$ as a function of the history $h_{t-1},$ and only then determine his
action $a_{t}.$ Thus, in particular, there cannot be deterministic
calibrated procedures (because there are no leaky calibrated procedures);
see Dawid (1985) and Oakes (1985).

In the case of smooth calibration, the procedure that we construct is
deterministic, and thus leaky smoothly calibrated. However, there are also
randomized leaky smooothly calibrated procedures. One example is the simple
calibrated procedure of D. Foster (1999) in the one-dimensional case (where $%
A=\{$\textquotedblleft rain"~, \textquotedblleft no rain"$\}$ and $C=[0,1]$%
): the forecast there is \textquotedblleft almost deterministic," in the
sense that the randomization is only between two very close forecasts (such
as $29.99\%$ and $30.01\%)$, and so can be shown to be leaky smoothly
calibrated. For another example, see footnote \ref{ftn:random-fixedpoint} in
Section \ref{s:weak calibration} below.

A particular instance of the leaky setup is one where the A-player uses a
fixed reaction function $g:C\rightarrow A$ that is a continuous mapping of
forecasts to actions; thus, $a_{t}=g(c_{t})$ (independently of time $t$ and
history $h_{t-1}$). In this case, leaky smooth calibration implies that most
of the forecasts that are used must be approximate fixed points of $g;$
indeed, in every period in which the forecast is $c$ the action is the same,
namely, $g(c),$ and so the average of the actions in all the periods where
the forecast is (close to) $c$ is (close to) $g(c)$ (use the continuity of $%
g $ here); formally, see the arguments in part (iv) of the proof of Theorem %
\ref{th:finite} in Section \ref{s:Nash}). Thus, leaky procedures find
(approximate) fixed points for arbitrary continuous functions $g$, and so
must in general be more complex than the procedures that yield calibration
(such as those obtained by Blackwell's approachability); cf. the complexity
class \textsc{PPAD} (Papadimitriou 1994) in the computer science literature
(see also Hazan and Kakade 2012 for the connection to calibration).

\subsection{Result\label{sus:results}}

A strategy $\sigma $ has \emph{finite recall }and is\emph{\ stationary} if
there exists a finite integer $R\geq 1$ and a function $\tilde{\sigma}%
:(C\times A)^{R}\rightarrow C$ such that 
\[
\sigma (h_{T-1})=\tilde{\sigma}%
(c_{T-R},a_{T-R},c_{T-R+1},a_{T-R+1},...,c_{T-1},a_{T-1}) 
\]%
for every $T>R$ and history $h_{T-1}=(c_{t},a_{t})_{1\leq t\leq T-1}.$ Thus,
only the \textquotedblleft window" consisting of the last $R$ periods
matters; the rest of the history, as well as the calendar time $T,$ do not.
Finally, a finite set $D\subseteq C$ is a $\delta $\emph{-grid} for $C$ if
for every $c\in C$ there is $d\equiv d(c)\in D$ such that $||d-c||\leq
\delta .$

Our result is:

\begin{theorem}
\label{th:smooth-c}For every $\varepsilon >0$ and $L<\infty $ there is an $%
(\varepsilon ,L)$-smoothly calibrated procedure. Moreover, the procedure may
be taken to be:

\begin{itemize}
\item deterministic;

\item leaky;

\item with finite recall and stationary; and

\item with all the forecasts lying on a finite grid.\footnote{%
The sizes $R$ of the recall and $\delta $ of the grid depend on $\varepsilon 
$, $L,$ the dimension $m,$ and the bound on the compact set $C.$}
\end{itemize}
\end{theorem}

The proof will proceed as follows. First, we construct deterministic
finite-recall algorithms for the \emph{online linear regression} problem
(cf. Foster 1991, Azoury and Warmuth 2001); see Theorem \ref{th:regression-R}
in Section \ref{s:linear regression}. Next, we use these algorithms to get
deterministic finite-recall \emph{weakly calibrated} procedures (cf. Foster
and Kakade 2004, 2006); see Theorem \ref{th:weak-calib} in Section \ref%
{s:weak calibration}. Finally, we obtain smooth calibration from weak
calibration; see Section \ref{s:smooth calibration}.

\section{Online Linear Regression\label{s:linear regression}}

Classical linear regression tries to predict a variable $y$ from a vector $x$
of $d$ variables (and so $y\in \mathbb{R}$ and $x\in \mathbb{R}^{d}).$ There
are observations $(x_{t},y_{t})_{t},$ and one typically assumes that%
\footnote{%
Vectors are viewed as column vectors, and $\theta ^{\prime }$ denotes the
transpose of $\theta $ (thus $\theta ^{\prime }x$ is the scalar product $%
\theta \cdot x$ of $\theta $ and $x).$} $y_{t}=\theta ^{\prime
}x_{t}+\epsilon _{t},$ where $\epsilon _{t}$ are (zero-mean normally
distributed) error terms. The optimal estimator for $\theta $ is then given
by the least squares method; i.e., $\theta $ minimizes $(1/T)\sum_{t=1}^{T}%
\psi _{t}(\theta )$ with 
\[
\psi _{t}(\theta ):=(y_{t}-\theta ^{\prime }x_{t})^{2} 
\]%
for every $t.$

In the \emph{online linear regression} problem (Foster 1991; see Section \ref%
{sus:literature}), the observations arrive sequentially, and at each time
period $t$ we want to determine $\theta _{t}$ given the information at that
time, namely, $(x_{1},y_{1}),...,(x_{t-1},y_{t-1})$ and $x_{t}$ only. The
goal is to bound the difference between the mean square errors in the online
case and the offline case (i.e., \textquotedblleft in hindsight"); namely,%
\[
\frac{1}{T}\sum_{t=1}^{T}\psi _{t}(\theta _{t})-\frac{1}{T}%
\sum_{t=1}^{T}\psi _{t}(\theta ). 
\]

Thus, an online linear-regression \emph{algorithm }takes as input a sequence 
$(x_{t},y_{t})_{t\geq 1}$ in $\mathbb{R}^{d}\times \mathbb{R}$ and gives as
output a sequence $(\theta _{t})_{t\geq 1}$ in $\mathbb{R}^{d},$ such that $%
\theta _{t}$ is a function only of $x_{1},y_{1},...,x_{t-1},y_{t-1},x_{t},$
for each $t.$

Our result is:

\begin{theorem}
\label{th:regression-R}Let $X,Y$ be positive reals, and $\varepsilon >0.$
Then there exists a positive integer $R_{0}\equiv R_{0}(\varepsilon ,X,Y,d)$
such that for every $R>R_{0}$ there is an $R$-recall stationary
deterministic algorithm that gives $(\theta _{t})_{t\geq 1}$ , such that%
\begin{eqnarray}
\frac{1}{R}\sum_{t=T-R+1}^{T}\left[ \psi _{t}(\theta _{t})-\psi _{t}(\theta )%
\right] &\leq &\varepsilon (1+\left\Vert \theta \right\Vert ^{2})\text{\ \
and}  \label{eq:regr-R-R} \\
\frac{1}{T}\sum_{t=1}^{T}\left[ \psi _{t}(\theta _{t})-\psi _{t}(\theta )%
\right] &\leq &\varepsilon (1+\left\Vert \theta \right\Vert ^{2})
\label{eq:regression-R}
\end{eqnarray}%
hold for every $T\geq R,$ every $\theta \in \mathbb{R}^{d},$ and every
sequence $(x_{t},y_{t})_{t\geq 1}$ in $\mathbb{R}^{d}\times \mathbb{R}$ with 
$\left\Vert x_{t}\right\Vert \leq X$ and $\left\vert y_{t}\right\vert \leq Y$
for all $t.$
\end{theorem}

When in addition $\theta $ is bounded,\footnote{%
For example, $\theta $ lies in the unit simplex of $\mathbb{R}^{d}.$} say, $%
\left\Vert \theta \right\Vert \leq M,$ the mean square error of our online
algorithm is guaranteed not to exceed the optimal offline mean square error
by more than $\varepsilon (1+M^{2}).$

The proof of Theorem \ref{th:regression-R} in the remainder of this section
proceeds as follows.

\begin{description}
\item \emph{(i) Forward algorithm} (Section \ref{sus:FA}). We start with the
\textquotedblleft forward algorithm" of Azoury and Warmouth (2001) and the
resulting bound on the sum of regrets $\psi _{t}(\theta _{t})-\psi
_{t}(\theta )$ (Theorem \ref{th:A&W}).

\item \emph{(ii) Discounted forward algorithm }(Section \ref{sus:DFA}).\emph{%
\ }We modify the procedure by introducing a $\lambda $-discount factor,
which gives bounds on the discounted sum of regrets (Propositions \ref{p:DFA}
and \ref{p:DFA1}).

\item \emph{(iii) Windowed discounted forward algorithm} (Section \ref%
{sus:WDFA}). We further modify the procedure by restricting the history to a
window consisting only of the last $R$ periods, which gives bounds on the
sum of regrets over that window (Proposition \ref{p:WDFA}).

\item \emph{(iv) Conclusion} (Section \ref{sus:WDFA}).\emph{\ }One of the
useful properties of discounting is that the weight of the initial segment
from $1$ up to $T-R$ is small relative to the whole sum from $1$ to $T,$ and
so dropping that initial segment has little effect on the procedure and the
resulting estimates. We can thus choose an appropriate discount factor $%
\lambda $ and a window size $R$ yielding the desired bounds on the windowed
sum of regrets, and thus also on the time average of the regrets
(Proposition \ref{p:WDFA2}, which yields Theorem \ref{th:regression-R}).
\end{description}

\subsection{Forward Algorithm\label{sus:FA}}

The starting point is the following algorithm of Azoury and Warmuth (2001,
Section 5.4). For each $a>0,$ the $a$\emph{-forward algorithm} gives%
\footnote{$Z_{t}^{-1}$ is the inverse of the $d\times d$ matrix $Z_{t}$
(which is invertible because $a>0),$ and $I$ denotes the identity matrix.} $%
\theta _{t}=Z_{t}^{-1}v_{t},$ where%
\begin{equation}
Z_{t}=aI+\sum_{q=1}^{t}x_{q}x_{q}^{\prime }\;\;\;\;\text{and\ \ \ \ }%
v_{t}=\sum_{q=1}^{t-1}y_{q}x_{q}.  \label{eq:FA}
\end{equation}

\begin{theorem}[Azoury and Warmuth 2001]
\label{th:A&W}For every $a>0,$ the $a$-forward algorithm yields 
\begin{equation}
\sum_{t=1}^{T}\left[ \psi _{t}(\theta _{t})-\psi _{t}(\theta )\right] \leq
a\left\Vert \theta \right\Vert ^{2}+\sum_{t=1}^{T}y_{t}^{2}\left( 1-\frac{%
\det (Z_{t-1})}{\det (Z_{t})}\right)  \label{eq:5.17}
\end{equation}%
for every $T\geq 1,$ every $\theta \in \mathbb{R}^{d},$ and every sequence $%
(x_{t},y_{t})_{t\geq 1}$ in $\mathbb{R}^{d}\times \mathbb{R}.$
\end{theorem}

\begin{proof}
Theorem 5.6 and Lemma A.1 in Azoury and Warmuth (2001), where $Z_{t}$
denotes their $\eta _{t}^{-1}$ matrix; the second term in their formula
(5.17) is nonnegative since $\eta _{t}$ is a positive definite matrix.%
\footnote{%
Our statement is different from theirs because $\psi _{t}$ equals twice $%
L_{t},$ and there is a misprinted sign in the first line of their formula
(5.17).}
\end{proof}

\subsection{Discounted Forward Algorithm\label{sus:DFA}}

Let $a>0$ and $0<\lambda <1.$ The $\lambda $\emph{-discounted }$a$-\emph{%
forward algorithm} gives $\theta _{t}=Z_{t}^{-1}v_{t},$ where%
\begin{equation}
Z_{t}=aI+\sum_{q=1}^{t}\lambda ^{t-q}x_{q}x_{q}^{\prime }\;\;\;\;\text{and\
\ \ \ }v_{t}=\sum_{q=1}^{t-1}\lambda ^{t-q}y_{q}x_{q}.  \label{eq:DFA}
\end{equation}

\begin{proposition}
\label{p:DFA}For every $a>0$ and $0<\lambda <1,$ the $\lambda $-discounted $%
a $-forward algorithm yields%
\begin{equation}
\sum_{t=1}^{T}\lambda ^{T-t}\left[ \psi _{t}(\theta _{t})-\psi _{t}(\theta )%
\right] \leq a\left\Vert \theta \right\Vert ^{2}+\sum_{t=1}^{T}\lambda
^{T-t}y_{t}^{2}\left( 1-\lambda ^{d}\frac{\det (Z_{t-1})}{\det (Z_{t})}%
\right)  \label{eq:est-DFA}
\end{equation}%
for every $T\geq 1,$ every $\theta \in \mathbb{R}^{d},$ and every sequence $%
(x_{t},y_{t})_{t\geq 1}\ $in $\mathbb{R}^{d}\times \mathbb{R}.$
\end{proposition}

\begin{proof}
Let $b:=\sqrt{a(1-\lambda )}$. From the sequence $(x_{t},y_{t})_{t\geq 1}$
we construct a sequence $(\tilde{x}_{s},\tilde{y}_{s})_{s\geq 1}$ in blocks
as follows. For every $t\geq 1,$ the $t$-th block $B_{t}$ is of size $d+1$
and consists of $(\lambda ^{-t/2}be^{(1)},0),...,(\lambda
^{-t/2}be^{(d)},0),(\lambda ^{-t/2}x_{t},\lambda ^{-t/2}y_{t}),$ where $%
e^{(i)}$ is the $i$-th unit vector in $\mathbb{R}^{d}.$ The $a$-forward
algorithm applied to $(\tilde{x}_{s},\tilde{y}_{s})_{s\geq 1}$ yields the
following.

For $s=(d+1)t,$ i.e., at the end of the $B_{t}$ block, we have\footnote{%
The notation $\tilde{Z}_{s},$ $\tilde{\psi}_{s},...$ pertains to the $(%
\tilde{x}_{s},\tilde{y}_{s})_{s\geq 1}$ problem.} $\sum_{s\in B_{t}}\tilde{x}%
_{s}\tilde{x}_{s}^{\prime }=b^{2}\lambda ^{-t}\sum_{i=1}^{d}e^{(i)}\left(
e^{(i)}\right) ^{\prime }+\lambda ^{-t}x_{t}x_{t}^{\prime }=\lambda
^{-t}(b^{2}I+x_{t}x_{t}^{\prime });$ thus%
\begin{eqnarray*}
\tilde{Z}_{(d+1)t} &=&aI+\sum_{q=1}^{t}\sum_{s\in B_{t}}\tilde{x}_{s}\tilde{x%
}_{s}^{\prime }=aI+\sum_{q=1}^{t}\lambda ^{-q}(b^{2}I+x_{q}x_{q}^{\prime })
\\
&=&\lambda ^{-t}\left( aI+\sum_{q=1}^{t}\lambda ^{t-q}x_{q}x_{q}^{\prime
}\right) =\lambda ^{-t}Z_{t}
\end{eqnarray*}%
(since $\sum_{i=1}^{d}e^{(i)}\left( e^{(i)}\right) ^{\prime }=I$ and $%
b^{2}=(1-\lambda )a;$ recall (\ref{eq:DFA})). Together with $\tilde{v}%
_{(d+1)t}=\sum_{q=1}^{t}\sum_{s\in B_{t}}\tilde{y}_{s}\tilde{x}%
_{s}=\sum_{q=1}^{t}\lambda ^{-q}y_{q}x_{q}=\lambda ^{-t}v_{t}$ (only the
first entry in each block has a nonzero $\tilde{y}),$ it follows that $%
\tilde{\theta}_{(d+1)t}$ indeed equals $\theta _{t}=Z_{t}^{-1}v_{t}$ as
given by (\ref{eq:DFA}).

Next, for every $t$ we have $\sum_{s\in B_{t}}\tilde{\psi}_{s}(\tilde{\theta}%
_{s})\geq \lambda ^{-t}\psi _{t}(\theta _{t})$ (all terms in the sum are
nonnegative, and we drop all except the last one)$.$ Also, for every $\theta
\in \mathbb{R}^{d},$ 
\[
\sum_{s\in B_{t}}\tilde{\psi}_{s}(\theta )=\lambda ^{-t}\left(
b^{2}\sum_{i=1}^{d}(\theta ^{\prime }e^{(i)})^{2}+\psi _{t}(\theta )\right)
=\lambda ^{-t}\left( b^{2}\left\Vert \theta \right\Vert ^{2}+\psi
_{t}(\theta )\right) . 
\]%
Thus the left-hand side of (\ref{eq:5.17}) evaluated at the end of the $T$%
-th block $B_{T}$ satisfies%
\begin{eqnarray*}
LHS &\geq &\sum_{t=1}^{T}\lambda ^{-t}\psi _{t}(\theta _{t})-a\left\Vert
\theta \right\Vert ^{2}-b^{2}\left\Vert \theta \right\Vert
^{2}\sum_{t=1}^{T}\lambda ^{-t}-\sum_{t=1}^{T}\lambda ^{-t}\psi _{t}(\theta )
\\
&=&\sum_{t=1}^{T}\lambda ^{-t}\left[ \psi _{t}(\theta _{t})-\psi _{t}(\theta
)\right] -\lambda ^{-T}a\left\Vert \theta \right\Vert ^{2}.
\end{eqnarray*}%
On the right-hand side we get%
\[
RHS=\sum_{t=1}^{T}\sum_{s\in B_{t}}\tilde{y}_{s}^{2}\left( 1-\frac{\det (%
\tilde{Z}_{s-1})}{\det (\tilde{Z}_{s})}\right) =\sum_{t=1}^{T}\lambda
^{-t}y_{t}^{2}\left( 1-\frac{\det (\tilde{Z}_{(d+1)t-1})}{\det (\tilde{Z}%
_{(d+1)t})}\right) 
\]%
(again, only the last term in each block has nonzero $\tilde{y}_{s}).$ We
have seen above that $\tilde{Z}_{(d+1)t}=\lambda ^{-t}Z_{t};$ thus $\tilde{Z}%
_{(d+1)t-1}=\tilde{Z}_{(d+1)t}-\lambda ^{-t}x_{t}x_{t}^{\prime }=\lambda
^{-t}(Z_{t}-x_{t}x_{t}^{\prime })=\lambda ^{-t+1}Z_{t-1}+(\lambda
^{-t}-\lambda ^{-t+1})aI.$ Therefore $\det (\tilde{Z}_{(d+1)t-1})\geq \det
(\lambda ^{-t+1}Z_{t-1})$ (indeed, if $B$ is a positive definite matrix and $%
\beta >0$ then\footnote{%
Let $\beta _{1},...,\beta _{d}>0$ be the eigenvalues of $B;$ then the
eigenvalues of $B+cI$ are $\beta _{1}+c,...,\beta _{d}+c,$ and so $\det
(B+cI)=\prod_{i}(\beta _{i}+c)>\prod_{i}\beta _{i}=\det (B).$} $\det
(B+\beta I)>\det (B)).$ Therefore we obtain 
\begin{eqnarray*}
RHS &\leq &\sum_{t=1}^{T}\lambda ^{-t}y_{t}^{2}\left( 1-\frac{\det (\lambda
^{-t+1}Z_{t-1})}{\det (\lambda ^{-t}Z_{t})}\right) \\
&=&\sum_{t=1}^{T}\lambda ^{-t}y_{t}^{2}\left( 1-\lambda ^{d}\frac{\det
(Z_{t-1})}{\det (Z_{t})}\right)
\end{eqnarray*}%
(the matrices $Z_{t}$ are of size $d\times d,$ and so $\det
(cZ_{t})=c^{d}\det (Z_{t})).$ Recalling that $LHS\leq RHS$ by (\ref{eq:5.17}%
) and multiplying by $\lambda ^{T}$ yields the result.
\end{proof}

\bigskip

\noindent \textbf{Remark. }From now on it will be convenient to assume that $%
\left\Vert x_{t}\right\Vert \leq 1$ and $\left\vert y_{t}\right\vert \leq 1$
(i.e., $X=Y=1);$ for general $X$ and $Y,$ multiply $x_{t},y_{t},\theta
_{t},\psi _{t},a$ by $X,Y,Y/X,Y^{2},X^{2},$ respectively, in the appropriate
formulas.

\begin{proposition}
\label{p:DFA1}For every $a>0$ and $1/4\leq \lambda <1$ there exists a
constant $D_{1}\equiv D_{1}(a,\lambda ,d)$ such that the $\lambda $%
-discounted $a$-forward algorithm yields%
\begin{equation}
\sum_{t=1}^{T}\lambda ^{T-t}\left[ \psi _{t}(\theta _{t})-\psi _{t}(\theta )%
\right] \leq a\left\Vert \theta \right\Vert ^{2}+D_{1},  \label{eq:est-DFA1}
\end{equation}%
for every $T\geq 1,$ every $\theta \in \mathbb{R}^{d},$ and every sequence $%
(x_{t},y_{t})_{t\geq 1}$ in $\mathbb{R}^{d}\times \mathbb{R}$ with $%
\left\Vert x_{t}\right\Vert \leq 1$ and $|y_{t}|\leq 1$ for all $t.$
\end{proposition}

\begin{proof}
Let $K\geq 1$ be an integer such that $1/4\leq \lambda ^{K}\leq 1/2.$ Given $%
T\geq 1,$ let the integer $m\geq 1$ satisfy $(m-1)K<T\leq mK.$ Writing $%
\zeta _{t}$ for $\det (Z_{t}),$ we have%
\begin{eqnarray}
\sum_{t=1}^{T}\lambda ^{T-t}\left( 1-\lambda ^{d}\frac{\zeta _{t-1}}{\zeta
_{t}}\right) &\leq &\sum_{t=1}^{mK}\lambda ^{T-t}\left( 1-\lambda ^{d}\frac{%
\zeta _{t-1}}{\zeta _{t}}\right)  \nonumber \\
&=&\lambda ^{T}\sum_{j=1}^{m-1}\sum_{t=jK+1}^{(j+1)K}\lambda ^{-t}\ln \left(
\lambda ^{-d}\frac{\zeta _{t}}{\zeta _{t-1}}\right)  \nonumber \\
&\leq &\lambda ^{T}\sum_{j=1}^{m-1}\lambda
^{-(j+1)K}\sum_{t=jK+1}^{(j+1)K}\ln \left( \lambda ^{-d}\frac{\zeta _{t}}{%
\zeta _{t-1}}\right)  \nonumber \\
&\leq &\lambda ^{T}\sum_{j=1}^{m-1}\lambda ^{-(j+1)K}\ln \left( \lambda
^{-dK}\frac{\zeta _{(j+1)K}}{\zeta _{jK}}\right)  \label{eq:K}
\end{eqnarray}%
(in the second line we have used $1-1/u\leq \ln u$ for $0<u\leq 1$, as in
(4.21) in Azoury and Warmuth 2001; in the third line, $\lambda ^{-t}\leq
\lambda ^{-(j+1)K}$ since $t\leq (j+1)K$ and $\lambda <1$).

Let $B=(b_{ij})$ be a $d\times d$ symmetric positive definite matrix with $%
\left\vert b_{ij}\right\vert \leq \beta $ for all $i,j,$ and let $a>0.$ Then 
$a^{d}\leq \det (aI+B)\leq d!(a+\beta )^{d}.$ Indeed, the second inequality
follows easily since the determinant is the sum of $d!$ products of $d$
elements each. For the first inequality, let $\beta _{1},...,\beta _{d}>0$
be the eigenvalues of $B;$ then the eigenvalues of $aI+B$ are $a+\beta
_{1},...,a+\beta _{d},$ and so $\det (aI+B)=\Pi _{i=1}^{d}(a+\beta
_{i})>a^{d}$. Applying this to $Z_{t}$ (using (\ref{eq:DFA}), $%
|x_{t,i}x_{t,j}|\leq \left\Vert x_{t}\right\Vert ^{2}\leq 1,$ and $%
\sum_{t=1}^{T}\lambda ^{T-t}<1/(1-\lambda )$) yields 
\[
a^{d}\leq \zeta _{t}\equiv \det (Z_{t})\leq d!\left( a+\frac{1}{1-\lambda }%
\right) ^{d}. 
\]%
Therefore, since $\lambda ^{-K}\leq 4,$ we get 
\[
\lambda ^{-dK}\frac{\zeta _{(j+1)K}}{\zeta _{jK}}\leq 4^{d}d!\left( 1+\frac{1%
}{a(1-\lambda )}\right) ^{d}=:D, 
\]%
and so (\ref{eq:K}) is%
\begin{eqnarray*}
&\leq &\lambda ^{T}\sum_{j=1}^{m-1}\left( \lambda ^{-K}\right) ^{j+1}\ln
D\leq \lambda ^{T}\frac{\lambda ^{-K(m+1)}-\lambda ^{-2K}}{\lambda ^{-K}-1}%
\ln D \\
&\leq &\lambda ^{T}\frac{\lambda ^{-T}\lambda ^{-K}-0}{2-1}\ln D=4\ln D
\end{eqnarray*}%
(since $2\leq \lambda ^{-K}\leq 4$ and $K(m+1)<T+K$). Substituting this in (%
\ref{eq:est-DFA}) and putting 
\begin{equation}
D_{1}:=4\left( \ln d!+d\ln 4+d\ln \left( 1+\frac{1}{a(1-\lambda )}\right)
\right)  \label{eq:D1}
\end{equation}%
completes the proof.
\end{proof}

\subsection{Windowed Discounted Forward Algorithm\label{sus:WDFA}}

From now on it is convenient to put $(x_{t},y_{t},\theta _{t})=(0,0,0)$ for
all $t\leq 0.$

Let $a>0,~0<\lambda <1,$ and integer $R\geq 1$. The $R$\emph{-windowed }$%
\lambda $-\emph{discounted }$a$-\emph{forward algorithm} gives $\theta
_{t}=Z_{t}^{-1}v_{t},$ where\footnote{%
The sums below effectively start at $\min \{t-R+1,1\}$ (because we put $%
x_{q}=0$ for $q\leq 0$).}%
\begin{equation}
Z_{t}=aI+\sum_{q=t-R+1}^{t}\lambda ^{t-q}x_{q}x_{q}^{\prime }\;\;\;\;\text{%
and\ \ \ }v_{t}=\sum_{q=t-R+1}^{t-1}\lambda ^{R-q}y_{q}x_{q}.
\label{eq:RDFA}
\end{equation}

\begin{lemma}
\label{p:theta-theta(R)}For every $a>0$ and $0<\lambda <1$ there exists a
constant $D_{2}\equiv D_{2}(a,\lambda ,d)$ such that if $(\tilde{\theta}%
_{t})_{t\geq 1}$ is given by the $\lambda $-discounted $a$-forward
algorithm, and $(\theta _{t})_{t\geq 1}$ is given by the $R$-windowed $%
\lambda $-discounted $a$-forward algorithm for some integer $R\geq 1,$ then 
\begin{equation}
\left\vert \psi _{t}(\tilde{\theta}_{t})-\psi _{t}(\theta _{t})\right\vert
\leq D_{2}\lambda ^{R}  \label{eq:theta-theta(R)}
\end{equation}%
for every\footnote{%
For $t\leq R$ we have $\tilde{\theta}_{t}=\theta _{t}$ since they are given
by the same formula.} $t\geq 1$ and every sequence $(x_{t},y_{t})_{t\geq 1}$
in $\mathbb{R}^{d}\times \mathbb{R}$ with $\left\Vert x_{t}\right\Vert \leq
1 $ and $|y_{t}|\leq 1$ for all $t$.
\end{lemma}

To prove this lemma we use the following basic result. The norm of a matrix $%
A$ is $\left\Vert A\right\Vert :=\max {}_{z\neq 0}\left\Vert Az\right\Vert
/\left\Vert z\right\Vert .$

\begin{lemma}
\label{l:sensitivity}For $k=1,2,$ let $c_{k}=A_{k}^{-1}b_{k},$ where $A_{k}$
is a $d\times d$ symmetric matrix whose eigenvalues are all greater than or
equal to some $\alpha >0,$ and $\left\Vert b_{k}\right\Vert \leq M.$ Then $%
\left\Vert c_{k}\right\Vert \leq M/\alpha $ and%
\[
\left\Vert c_{1}-c_{2}\right\Vert \leq \frac{1}{\alpha }\left\Vert
b_{1}-b_{2}\right\Vert +\frac{M}{\alpha ^{2}}\left\Vert
A_{1}-A_{2}\right\Vert . 
\]
\end{lemma}

\begin{proof}
First, $\left\Vert c_{k}\right\Vert \leq \left\Vert A_{k}^{-1}\right\Vert
\left\Vert b_{k}\right\Vert \leq (1/\alpha )M$ since $\left\Vert
A_{k}^{-1}\right\Vert $ equals the maximal eigenvalue of $A_{k}^{-1},$ which
is the reciprocal of the minimal eigenvalue of $A_{k},$ and so $\left\Vert
A_{k}^{-1}\right\Vert \leq 1/\alpha .$

Second, express $c_{1}-c_{2}$ as $%
A_{1}^{-1}(b_{1}-b_{2})+A_{1}^{-1}(A_{2}-A_{1})A_{2}^{-1}b_{2},$ to get%
\[
\left\Vert c_{1}-c_{2}\right\Vert \leq \left\Vert A_{1}^{-1}\right\Vert
\left\Vert b_{1}-b_{2}\right\Vert +\left\Vert A_{1}^{-1}\right\Vert
\left\Vert A_{2}-A_{1}\right\Vert \left\Vert A_{2}^{-1}\right\Vert
\left\Vert b_{2}\right\Vert 
\]%
and the proof is complete.
\end{proof}

\bigskip

\begin{proof}[Proof of Lemma \protect\ref{p:theta-theta(R)}]
For $t\leq R$ we have $\tilde{\theta}_{t}\equiv \theta _{t},$ and so
consider $t>R.$ We have\footnote{%
Notation: $\tilde{v}_{t}$ and $\tilde{Z}_{t}$ pertain to the sequence $%
\tilde{\theta}_{t}$ given by the $\lambda $-discounted $a$-forward
algorithm, whereas $v_{t}$ and $Z_{t}$ pertain to the sequence $\theta _{t}$
given by the $R$-windowed $\lambda $-discounted $a$-forward algorithm.} $%
\left\Vert \tilde{v}_{t}\right\Vert ,\left\Vert v_{t}\right\Vert \leq
\sum_{q=1}^{\infty }\lambda ^{q}=1/(1-\lambda )$. The matrices $\tilde{Z}%
_{t} $ and $Z_{t}$ are the sum of $aI$ and a positive-definite matrix, and
so their eigenvalues are $\geq a.$ Next, 
\[
\left\Vert \tilde{v}_{t}-v_{t}\right\Vert =\left\Vert
\sum_{q=1}^{t-R}\lambda ^{t-q}y_{q}x_{q}\right\Vert \leq \frac{\lambda ^{R}}{%
1-\lambda }; 
\]%
similarly, for each each element $(\tilde{Z}_{t}-Z_{t})_{ij}$ of $\tilde{Z}%
_{t}-Z_{t}$ we have 
\[
\left\vert (\tilde{Z}_{t}-Z_{t})_{ij}\right\vert =\left\vert
\sum_{q=1}^{t-R}\lambda ^{t-q}x_{q,i}x_{q,j}\right\vert \leq \frac{\lambda
^{R}}{1-\lambda }, 
\]%
and so\footnote{%
Because $\left\Vert A\right\Vert \leq d\max_{i,j}|a_{ij}|$ for any $d\times
d $ matrix $A.$} $\left\Vert \tilde{Z}_{t}-Z_{t}\right\Vert \leq d\lambda
^{R}/(\lambda -1).$ Using Lemma \ref{l:sensitivity} yields%
\[
\left\Vert \tilde{\theta}_{t}-\theta _{t}\right\Vert \leq \frac{1}{a}\frac{%
\lambda ^{R}}{1-\lambda }+\frac{1}{a^{2}}\frac{d\lambda ^{R}}{1-\lambda }=%
\frac{\lambda ^{R}(a+d)}{(1-\lambda )a^{2}}. 
\]%
Hence%
\begin{eqnarray*}
\left\vert \psi _{t}(\tilde{\theta}_{t})-\psi _{t}(\theta _{t})\right\vert
&=&\left\vert (y_{t}-\tilde{\theta}_{t}^{\prime }x_{t})^{2}-(y_{t}-\theta
_{t}^{\prime }x_{t})^{2}\right\vert \\
&=&\left\vert (\tilde{\theta}_{t}^{\prime }-\theta _{t}^{\prime })x_{t}\cdot
\left( 2y_{t}-(\tilde{\theta}_{t}^{\prime }+\theta _{t}^{\prime
})x_{t}\right) \right\vert \\
&\leq &\left\Vert \tilde{\theta}_{t}-\theta _{t}\right\Vert \left(
2+\left\Vert \tilde{\theta}_{t}\right\Vert +\left\Vert \theta
_{t}\right\Vert \right) \\
&\leq &\frac{\lambda ^{R}(a+d)}{(1-\lambda )a^{2}}\left( 2+\frac{2}{%
(1-\lambda )a}\right) =D_{2}\lambda ^{R},
\end{eqnarray*}%
where%
\begin{equation}
D_{2}:=\frac{2\left( a+d\right) (a(1-\lambda )+1)}{a^{3}(1-\lambda )^{2}};
\label{eq:D2}
\end{equation}%
this completes the proof.
\end{proof}

\begin{proposition}
\label{p:WDFA}For every $a>0$ and $1/4\leq \lambda <1$ there exist constants 
$D_{1}\equiv D_{1}(a,\lambda ,d)$ and $D_{2}\equiv D_{2}(a,\lambda ,d)$ such
that for every integer $R\geq 1$ the $R$-windowed $\lambda $-discounted $a$%
-forward algorithm yields%
\begin{eqnarray}
\frac{1}{R}\sum_{t=T-R+1}^{T}\left[ \psi _{t}(\theta _{t})-\psi _{t}(\theta )%
\right] &\leq &(a\left\Vert \theta \right\Vert ^{2}+D_{1})\left( 1-\lambda +%
\frac{\lambda }{R}\right)  \label{eq:est-WDFA} \\
&&+\frac{(\left\Vert \theta \right\Vert +1)^{2}}{R(1-\lambda )}+D_{2}\lambda
^{R}  \nonumber
\end{eqnarray}%
for every $T\geq 1,$ every $\theta \in \mathbb{R}^{d},$ and every sequence $%
(x_{t},y_{t})_{t\geq 1}$ in $\mathbb{R}^{d}\times \mathbb{R}$ with $%
\left\Vert x_{t}\right\Vert \leq 1$ and $\left\vert y_{t}\right\vert \leq 1$
for all $t.$
\end{proposition}

\begin{proof}
Let $\tilde{\theta}_{t}$ be given by the $\lambda $-discounted $a$-forward
algorithm. Put $g_{t}:=\psi _{t}(\theta _{t})-\psi _{t}(\theta )$ (where $%
\theta _{t}$ is given by the $R$-windowed $\lambda $-discounted $a$-forward
algorithm) and $\tilde{g}_{t}:=\psi _{t}(\tilde{\theta}_{t})-\psi
_{t}(\theta ).$ Apply (\ref{eq:est-DFA1}) at $T,$ and also at each one of $%
T-R+1,T-R+2,...,T-1;$ multiply those by $1-\lambda $ and add them all up, to
get 
\begin{eqnarray*}
\sum_{t=1}^{T}\lambda ^{T-t}\tilde{g}_{t}+(1-\lambda
)\sum_{r=1}^{R-1}\sum_{t=1}^{T-r}\lambda ^{T-r-t}\tilde{g}_{t} &\leq
&(a\left\Vert \theta \right\Vert ^{2}+D_{1})(1+(R-1)(1-\lambda )) \\
&=&(a\left\Vert \theta \right\Vert ^{2}+D_{1})(R-R\lambda +\lambda ).
\end{eqnarray*}%
For $t\leq T-R,$ the total coefficient of $\tilde{g}_{t}$ on the left-hand
side above is $\lambda ^{T-t}+(1-\lambda )\sum_{r=1}^{R-1}\lambda
^{T-r-t}=\lambda ^{T-R+1-t};$ for $T-R+1\leq t\leq T,$ it is $\lambda
^{T-t}+(1-\lambda )\sum_{r=1}^{T-t}\lambda ^{T-r-t}=1.$ Therefore%
\[
\sum_{t=1}^{T-R}\lambda ^{T-R+1-t}\tilde{g}_{t}+\sum_{t=T-R+1}^{T}\tilde{g}%
_{t}\leq (a\left\Vert \theta \right\Vert ^{2}+D_{1})(R-R\lambda +\lambda ). 
\]%
Now $\tilde{g}_{t}\geq -\psi _{t}(\theta )\geq -(\left\Vert \theta
\right\Vert \left\Vert x_{t}\right\Vert +\left\vert y_{t}\right\vert
)^{2}\geq -(\left\Vert \theta \right\Vert +1)^{2},$ and so%
\begin{eqnarray*}
\sum_{t=T-R+1}^{T}\tilde{g}_{t} &\leq &(a\left\Vert \theta \right\Vert
^{2}+D_{1})(R-R\lambda +\lambda )+(\left\Vert \theta \right\Vert
+1)^{2}\sum_{t=1}^{T-R}\lambda ^{T-R+1-t} \\
&\leq &(a\left\Vert \theta \right\Vert ^{2}+D_{1})(R-R\lambda +\lambda )+%
\frac{(\left\Vert \theta \right\Vert +1)^{2}}{1-\lambda }.
\end{eqnarray*}%
Divide by $R$ and use $g_{t}\leq \tilde{g}_{t}+D_{2}\lambda ^{R}$ (by
Proposition \ref{p:theta-theta(R)}).
\end{proof}

\bigskip

Choosing appropriate $\lambda $ and $R$ allows us to bound the right-hand
side of (\ref{eq:est-WDFA}).

\begin{proposition}
\label{p:WDFA2}For every $\varepsilon >0$ and $a>0$ there is $\lambda
_{0}\equiv \lambda _{0}(\varepsilon ,a,d)<1$ such that for every $\lambda
_{0}<\lambda <1$ there is $R_{0}\equiv R_{0}(\varepsilon ,a,d,\lambda )\geq
1 $ such that for every $R\geq R_{0}$ the $R$-windowed $\lambda $-discounted 
$a $-forward algorithm yields%
\begin{eqnarray}
\frac{1}{R}\sum_{t=T-R+1}^{T}\left[ \psi _{t}(\theta _{t})-\psi _{t}(\theta )%
\right] &\leq &\varepsilon (1+\left\Vert \theta \right\Vert ^{2}),
\label{eq:R-R} \\
\frac{1}{T}\sum_{t=1}^{T}\left[ \psi _{t}(\theta _{t})-\psi _{t}(\theta )%
\right] &\leq &\varepsilon (1+\left\Vert \theta \right\Vert ^{2}),
\label{eq:R-T}
\end{eqnarray}%
for every $T\geq R,$ every $\theta \in \mathbb{R}^{d},$ and every sequence $%
(x_{t},y_{t})_{t\geq 1}$ in $\mathbb{R}^{d}\times \mathbb{R}$ with $%
\left\Vert x_{t}\right\Vert \leq 1$ and $\left\vert y_{t}\right\vert \leq 1$
for all $t.$
\end{proposition}

\begin{proof}
The right-hand side of (\ref{eq:est-WDFA}) is%
\[
\leq \left( D_{1}(1-\lambda )+\frac{D_{1}}{R}+\frac{2}{R(1-\lambda )}%
+D_{2}\lambda ^{R}\right) +\left\Vert \theta \right\Vert ^{2}\left(
a(1-\lambda )+\frac{a}{R}+\frac{2}{R(1-\lambda )}\right) 
\]%
(use $\lambda /R\leq 1/R$ and $(\left\Vert \theta \right\Vert +1)^{2}\leq
2\left\Vert \theta \right\Vert ^{2}+2).$ First, take $1/4\leq \lambda _{0}<1$
close enough to $1$ so that $a(1-\lambda _{0})\leq \varepsilon /4$ and $%
D_{1}(a,\lambda _{0},d)\cdot (1-\lambda _{0})\leq \varepsilon /4$ (recall
formula (\ref{eq:D1}) for $D_{1}$ and use $\lim_{x\rightarrow 0^{+}}x\ln
x=0).$ Then, given $\lambda \in \lbrack \lambda _{0},1),$ take $R_{0}\geq 1\ 
$large enough so that $a/R_{0}\leq \varepsilon /4,$\ $D_{1}(a,\lambda
,d)/R_{0}\leq \varepsilon /4,$ $2/(R_{0}(1-\lambda ))\leq \varepsilon /4,$
and $D_{2}(a,\lambda ,d)\lambda ^{R_{0}}\leq \varepsilon /4.$ This shows (%
\ref{eq:R-R}) for every $T\geq 1.$

In particular, for $T^{\prime }<R$ we get $(1/R)\sum_{t=1}^{T^{\prime }}%
\left[ \psi _{t}(\theta _{t})-\psi _{t}(\theta )\right] \leq \varepsilon
(1+\left\Vert \theta \right\Vert ^{2})$ (because $(x_{t},y_{t},\theta
_{t})=(0,0,0)$ for all $t\leq 0).$ For $T\geq R,$ add up the inequalities (%
\ref{eq:R-R}) for the disjoint blocks of size $R$ that end at $t=T,$
together with the above inequality for the initial smaller block of size $%
T^{\prime }<R$ when $T$ is not a multiple of $R,$ to get $(1/R)\sum_{t=1}^{T}%
\left[ \psi _{t}(\theta _{t})-\psi _{t}(\theta )\right] \leq \lceil
T/R\rceil \varepsilon (1+\left\Vert \theta \right\Vert ^{2})\leq
2(T/R)\varepsilon (1+\left\Vert \theta \right\Vert ^{2}).$ Replacing $%
\varepsilon $ with $\varepsilon /2$ yields (\ref{eq:R-T}).
\end{proof}

\bigskip

\noindent \textbf{Remark. }Similar arguments show that, for $\lambda
_{0}\leq \lambda <1,$ the discounted average is also small:%
\[
\frac{1-\lambda }{1-\lambda ^{T}}\sum_{t=1}^{T}\lambda ^{T-t}\left[ \psi
_{t}(\theta _{t})-\psi _{t}(\theta )\right] \leq \varepsilon (1+\left\Vert
\theta \right\Vert ^{2}). 
\]

\bigskip

Proposition \ref{p:WDFA2} yields the main result of this section, Theorem %
\ref{th:regression-R}.

\bigskip

\begin{proof}[Proof of Theorem \protect\ref{th:regression-R}]
Use Proposition \ref{p:WDFA2} (with, say, $a=1),$ and rescale everything by $%
X$ and $Y$ appropriately (see the Remark before Proposition \ref{p:DFA1}).
\end{proof}

\section{Weak Calibration\label{s:weak calibration}}

The notion of \textquotedblleft weak calibration" was introduced by Kakade
and Foster (2004) and Foster and Kakade (2006). The idea is as follows.
Given a \textquotedblleft test" function $w:C\rightarrow \{0,1\}$ that
indicates which forecasts $c$ to consider, let the corresponding score be%
\footnote{%
The $S_{T}$ scores are norms of averages, rather than averages of norms like
the $K_{T}$ scores. \textquotedblleft Windowed" versions of the scores may
also be considered (with the average taken over the last $R$ periods only;
cf. (\ref{eq:regr-R-R})).} $S_{T}^{w}:=||(1/T)%
\sum_{t=1}^{T}w(c_{t})(a_{t}-c_{t})||.$ It can be shown that if $S_{T}^{w}$
is small for every such $w,$ then the calibration score $K_{T}$ is also
small.\footnote{%
Specifically, if $S_{T}^{w}\leq \varepsilon $ for all $w:C\rightarrow
\{0,1\} $ then $K_{T}\leq 2m\varepsilon .$ Indeed, for each coordinate $%
i=1,...,m,$ let $C_{+}^{i}$ be the set of all $c_{t}$ such that $\bar{a}%
_{t,i}>c_{t,i},$ and $C_{-}^{i}$ the set of all $c_{t}$ such that $\bar{a}%
_{t,i}<c_{t,i}.$ Taking $w$ to be the indicator of $C_{+}^{i}$ yields $%
S_{T}^{w}=(1/T)\sum_{t}[\bar{a}_{t,i}-c_{t,i}]_{+}\leq \varepsilon $ (where $%
[z]_{+}:=\max \{z,0\});$ similarly, the indicator of $C_{-}^{i}$ yields $%
(1/T)\sum_{t}[\bar{a}_{t,i}-c_{t,i}]_{-}\leq \varepsilon .$ Adding the two
inequalities gives $(1/T)\sum_{t}|\bar{a}_{t,i}-c_{t,i}|\leq 2\varepsilon .$
Since this holds for each one of the $m$ coordinates, it follows that $%
K_{T}\leq 2m\varepsilon .$}

Now instead of the discontinuous indicator functions, \emph{weak calibration}
requires that $S_{T}^{w}$ be small for Lipschitz continuous
\textquotedblleft weight" functions $w:C\rightarrow \lbrack 0,1]$.
Specifically, let $\varepsilon >0$ and $L<\infty .$ A procedure (i.e., a
strategy of the C-player in the calibration game) is $(\varepsilon ,L)$-%
\emph{weakly calibrated} if there is $T_{0}\equiv T_{0}(\varepsilon ,L)$
such that%
\begin{equation}
S_{T}^{w}=\left\Vert \frac{1}{T}\sum_{t=1}^{T}w(c_{t})(a_{t}-c_{t})\right%
\Vert \leq \varepsilon  \label{eq:(L,eps)-wc}
\end{equation}%
holds for every strategy of the A-player, every $T>T_{0},$ and every weight
function\ $w:C\rightarrow \lbrack 0,1]$ that is $L$-Lipschitz (i.e., $%
\mathcal{L}(w)\leq L).$

The importance of weak calibration is that, unlike regular calibration, it
can be guaranteed by deterministic procedures (which are thus leaky): Kakade
and Foster (2004) and Foster and Kakade (2006) have proven the existence of
deterministic $(\varepsilon ,L)$-weakly calibrated procedures. Moreover, as
we will show in the next section, weak calibration is essentially equivalent
to smooth calibration.

We now provide a deterministic $(\varepsilon ,L)$-weakly calibrated
procedure that in addition has finite recall and is stationary.

\begin{theorem}
\label{th:weak-calib}For every $\varepsilon >0$ and $L<\infty $ there exists
an $(\varepsilon ,L)$-weakly calibrated deterministic procedure that has
finite recall and is stationary; moreover, all its forecasts may be taken to
lie on a finite grid.
\end{theorem}

The proof uses the result of Theorem \ref{th:regression-R}. The basic idea
is to use the forecast itself as part of the input to the forecast---which
adds a fixed-point construct to the regression. Assume for starters that $a$
and $c$ are one-dimensional, and also that we have only a single weight
function $w$. Consider the online linear regression problem with $%
x_{t}=(c_{t},w(c_{t}))$ and $y_{t}=a_{t}.$ Given the history $h_{t-1}$ up to
and including time $t-1,$ if we knew the value of $c_{t}$ then we would get
a forecast $\hat{a}_{t}:=\theta _{t}^{\prime }x_{t}$ for which, by equation (%
\ref{eq:regression-R}), the regret is small. But we do not know $c_{t},$ as
it is going in fact to be our forecast: that is, we want to choose $c_{t}$
so that the resulting $\hat{a}_{t}$ satisfies $\hat{a}_{t}=c_{t}.$ This
requires solving a fixed-point problem (which is possible since the mapping $%
H$ from $c_{t}$ to $\hat{a}_{t}$ is continuous), and indeed yields the
desired $c_{t}.$ Now equation (\ref{eq:regression-R}) yields, for an
appropriate $\varepsilon >0,$%
\[
\frac{1}{T}\sum_{t=1}^{T}(a_{t}-\hat{a}_{t})^{2}\leq \frac{1}{T}%
\sum_{t=1}^{T}(a_{t}-\theta ^{\prime }x_{t})^{2}+\varepsilon 
\]%
for all $\theta $. But $\hat{a}_{t}=c_{t}$, and so taking $\theta =(1,\sqrt{%
\varepsilon })$ gives%
\begin{eqnarray*}
\frac{1}{T}\sum_{t=1}^{T}(a_{t}-c_{t})^{2} &\leq &\frac{1}{T}%
\sum_{t=1}^{T}(a_{t}-c_{t}-\sqrt{\varepsilon }w(c_{t}))^{2}+\varepsilon \\
&\leq &\frac{1}{T}\sum_{t=1}^{T}(a_{t}-c_{t})^{2}-2\sqrt{\varepsilon }\frac{1%
}{T}\sum_{t=1}^{T}w(c_{t})(a_{t}-c_{t})+\varepsilon +\varepsilon
\end{eqnarray*}%
(in the second line we have used $w(c_{t})\in \lbrack 0,1]).$ Therefore%
\[
\frac{1}{T}\sum_{t=1}^{T}w(c_{t})(a_{t}-c_{t})\leq \sqrt{\varepsilon }; 
\]%
together with the similar computation for $\theta =(1,-\sqrt{\varepsilon })$
we get%
\[
S_{T}^{w}=\left\vert \frac{1}{T}\sum_{t=1}^{T}w(c_{t})(a_{t}-c_{t})\right%
\vert \leq \sqrt{\varepsilon }, 
\]%
as desired. To deal with $m$-dimensional $a$ and $c$ we use $m$ separate
online regressions, one for each coordinate; to deal with all the $L$%
-Lipschitz weight functions $w,$ we take an appropriate finite grid.

\bigskip

\begin{proof}
\emph{(i) Preliminaries. }Without loss of generality assume that $A\subseteq
C\subseteq \lbrack 0,1]^{m}$ (one can always translate the sets $A$ and $C$%
---which does not affect (\ref{eq:(L,eps)-wc})---and rescale them---which
just rescales the Lipschitz constant); assume also that $L\geq 1$ (as $L$
increases there are more Lipschitz functions) and $\varepsilon \leq 1$.

For every $b\in \mathbb{R}^{m}$ let $\gamma (b):=\arg \min_{c\in
C}\left\Vert c-b\right\Vert $ be the closest point to $b$ in $C$ (it is well
defined and unique since $C$ is a convex compact set); then%
\begin{equation}
\left\Vert c-b\right\Vert \geq \left\Vert c-\gamma (b)\right\Vert
\label{eq:proj}
\end{equation}%
for every $c\in C$ (because 
\[
\left\Vert c-b\right\Vert ^{2}=\left\Vert c-\gamma (b)\right\Vert
^{2}+\left\Vert b-\gamma (b)\right\Vert ^{2}-2(b-\gamma (b))\cdot (c-\gamma
(b)) 
\]
and the third term is $\leq 0);$ moreover, $\gamma (b)=b$ when $b\in C.$

Let $\varepsilon _{1}:=\varepsilon /(2\sqrt{m}).$ Denote by $W_{L}$ the set
of weight functions $w:C\rightarrow \lbrack 0,1]$ with $\mathcal{L}(w)\leq
L. $ By Lemma \ref{l:lipschitz-basis} in the Appendix, there exist $d$
functions $f_{1},...,f_{d}$ in $W_{L}$ such that for every $w\in W_{L}$
there is a vector $\varpi \equiv \varpi _{w}\in \lbrack 0,1]^{d}$ with%
\footnote{\label{ft:dimension}Since $W_{L}$ is compact in the $\sup $ norm,
there are $f_{1},...,f_{d}\in W_{L}$ such that for every $w\in W_{L}$ there
is $1\leq i\leq d$ with $\max_{c\in C}|w(c)-f_{i}(c)|\leq \varepsilon _{1}.$
Lemma \ref{l:lipschitz-basis} improves on this, in getting a much smaller $d$
by using linear combinations with bounded coefficients.} 
\begin{equation}
\max_{c\in C}\left\vert w(c)-\sum_{i=1}^{d}\varpi _{i}f_{i}(c)\right\vert
\leq \varepsilon _{1}.  \label{eq:omega-w}
\end{equation}%
Denote $F(c):=(f_{1}(c),...,f_{d}(c))\in \lbrack 0,1]^{d};$ thus $\left\Vert
F(c)\right\Vert \leq \sqrt{d}.$ Without loss of generality we assume that
the set $\{f_{1},...,f_{d}\}$ includes the \textquotedblleft $j$-th
coordinate function," which maps each $c\in C$ to its $j$-th coordinate $%
c_{j};$ say, $f_{j}(c)=c_{j}$ for $j=1,...,m$ (thus $d>m;$ in fact $d$ is
much larger than $m).$

Let $\varepsilon _{2}:=\varepsilon /(m+m(1+d)^{2}+d^{2})$ (where $d$ is
given by Lemma \ref{l:lipschitz-basis} in the Appendix, and depends on $%
\varepsilon ,m,$ and $L)$ and $\varepsilon _{3}:=(\varepsilon _{2})^{2}.$

\emph{(ii) The function }$\emph{H}$\emph{. }Let $\lambda $ and $R$ be given
by Theorem \ref{th:regression-R} and Proposition \ref{p:WDFA2} for $a=1,$ $X=%
\sqrt{d},$ $Y=1,$ and $\varepsilon =\varepsilon _{3}$. For each $j=1,...,m$
consider the sequence $(x_{t},y_{t}^{(j)})_{t\geq
1}=(F(c_{t}),a_{t,j})_{t\geq 1}$ in $\mathbb{R}^{d}\times \mathbb{R},$ where 
$a_{t}\in A$ is determined by the A-player, and $c_{t}\in C$ is constructed
inductively as follows.

Let the history be $h_{t-1}=(c_{1},a_{1},...,c_{t-1},a_{t-1}).$ For each $%
c\in \mathbb{R}^{m}$, let\footnote{%
A subscript $j$ stands for the $j$-th coordinate (e.g., $a_{t,j}$ is the $j$%
-th coordinate of $a_{t})$, whereas a superscript $j$ refers to the $j$-th
procedure (e.g., $v_{t}^{(j)}$).} 
\begin{eqnarray*}
Z_{t}(c) &=&I+\sum_{q=1}^{R-1}\lambda ^{R-q}x_{q}x_{q}^{\prime
}+F(c)F(c)^{\prime }\in \mathbb{R}^{d\times d}, \\
v_{t}^{(j)} &=&\sum_{q=1}^{R-1}\lambda ^{R-q}a_{q,j}x_{q}\in \mathbb{R}^{d},
\\
H_{t,j}(c) &=&\left( Z_{t}(c)^{-1}v_{t}^{(j)}\right) ^{\prime }F(c)\in 
\mathbb{R}, \\
H_{t}(c) &=&(H_{t,1}(c),...,H_{t,m}(c))\in \mathbb{R}^{m}
\end{eqnarray*}%
(where $x_{q}=F(c_{q})$ for $q<t).$ Finally, we extend the function $H_{t}$
to all of $\mathbb{R}^{m}$ by putting $H_{t}(b):=H_{t}(\gamma (b))$ for
every $b\in \mathbb{R}^{m}$; i.e., we project $b$ to its closest point $%
\gamma (b)$ in $C,$ and then apply $H_{t}$ to it.

\emph{(iii) Fixed point of }$H.$ For every $c\in C$ we have $\left\Vert
v_{t}^{(j)}\right\Vert \leq \sqrt{d}\lambda /(1-\lambda )$ (since $%
\left\vert a_{q,j}\right\vert \leq 1$ and $\left\Vert x_{q}\right\Vert
=\left\Vert F(c_{q})\right\Vert \leq \sqrt{d}),$ and so $\left\Vert
Z_{t}(c)^{-1}v_{t}^{(j)}\right\Vert \leq \sqrt{d}\lambda /(1-\lambda )$ by
Lemma \ref{l:sensitivity} ($Z_{t}(c)$ is positive definite and its
eigenvalues are $\geq 1),$ which finally implies that $\left\vert
H_{t,j}(c)\right\vert \leq \sqrt{d}\lambda /(1-\lambda )\cdot \sqrt{d}%
=d\lambda /(1-\lambda )=:K.$ Therefore the restriction of $H_{t}$ to the
compact and convex set $[-K,K]^{m},$ which is clearly a continuous function
(since, again, $Z_{t}(c)$ is positive definite and its eigenvalues are $\geq
1),$ has a fixed point (by Brouwer's fixed-point theorem), which we denote $%
b_{t}$ (any fixed point will do);\footnote{\label{ftn:random-fixedpoint}%
There may be more than one fixed point here, in which case we may choose the
fixed point at random, and obtain a \emph{randomized} procedure that
satisfies everything the deterministic procedure does. Using it yields in
Theorem \ref{th:smooth-c} a randomized procedure that is \emph{leaky}
smoothly calibrated (cf. Section \ref{sus:leaky}).} put $c_{t}:=\gamma
(b_{t})\in C.$ Thus 
\[
c_{t}=\gamma (b_{t})\text{ \ and\ \ }b_{t}=H_{t}(b_{t})=H_{t}(c_{t}). 
\]%
Define $x_{t}:=F(\gamma (b_{t}))=F(c_{t})$ and $\theta
_{t}^{(j)}:=Z_{t}(c_{t})^{-1}v_{t}^{(j)}\in \mathbb{R}^{d}.$ Then $%
Z_{t}(c_{t})=I+\sum_{q=1}^{R}\lambda ^{R-q}x_{q}x_{q}^{\prime },$ and thus
it corresponds to the $R$-windowed $\lambda $-discounted $1$-forward
algorithm (see (\ref{eq:RDFA})). Therefore, for every $j=1,...,m$ and every $%
\theta ^{(j)}\in \mathbb{R}^{d}$ we have by (\ref{eq:regression-R})%
\begin{equation}
\frac{1}{T}\sum_{t=1}^{T}\left[ \psi _{t}^{(j)}(\theta _{t}^{(j)})-\psi
_{t}^{(j)}(\theta ^{(j)})\right] \leq \varepsilon _{3}\left( 1+\left\Vert
\theta ^{(j)}\right\Vert ^{2}\right)  \label{eq:theta(j)}
\end{equation}%
for all $T\geq T_{0}\equiv R,$ where $\psi _{t}^{(j)}(\theta
)=(a_{t,j}-\theta ^{\prime }x_{t})^{2},$ and thus $\psi _{t}^{(j)}(\theta
_{t}^{(j)})=(a_{t,j}-b_{t,j})^{2}$ (recall that $b_{t,j}=H_{t,j}(b_{t})=%
\left( \theta _{t}^{(j)}\right) ^{\prime }F(\gamma (b_{t}))=\left( \theta
_{t}^{(j)}\right) ^{\prime }x_{t}).$ Summing over $j$ yields%
\[
\frac{1}{T}\sum_{t=1}^{T}\sum_{j=1}^{m}\left[ \psi _{t}^{(j)}(\theta
_{t}^{(j)})-\psi _{t}^{(j)}(\theta ^{(j)})\right] \leq \varepsilon
_{3}\left( m+\sum_{j=1}^{m}\left\Vert \theta ^{(j)}\right\Vert ^{2}\right) . 
\]%
Now $\sum_{j=1}^{m}\psi _{t}^{(j)}(\theta
_{t}^{(j)})=\sum_{j=1}^{m}(a_{t,j}-b_{t,j})^{2}=\left\Vert
a_{t}-b_{t}\right\Vert ^{2}\geq \left\Vert a_{t}-\gamma (b_{t})\right\Vert
^{2}=\left\Vert a_{t}-c_{t}\right\Vert ^{2}=\sum_{j=1}^{m}$ $%
(a_{t,j}-c_{t,j})^{2}$ (by the definition of $\gamma (b_{t})$ and (\ref%
{eq:proj}), since $a_{t}\in A\subseteq C),$ and therefore%
\begin{equation}
\frac{1}{T}\sum_{t=1}^{T}\sum_{j=1}^{m}\left[ (a_{t,j}-c_{t,j})^{2}-\psi
_{t}^{(j)}(\theta ^{(j)})\right] \leq \varepsilon _{3}\left(
m+\sum_{j=1}^{m}\left\Vert \theta ^{(j)}\right\Vert ^{2}\right) .
\label{eq:all-j}
\end{equation}

\emph{(iv) Estimating }$S_{T}^{w}.$ Given a weight function $w\in W_{L},$
let the vector $\varpi \equiv \varpi _{w}\in \lbrack 0,1]^{d}$ satisfy (\ref%
{eq:omega-w}), i.e., $\left\vert w(c)-\varpi ^{\prime }F(c)\right\vert \leq
\varepsilon _{1}$ for all $c\in C.$ Take $u=(u_{j})_{j=1,...,m}\in \mathbb{R}%
^{m}$ with $\left\Vert u\right\Vert =1.$ For every $j=1,...,m,$ take $\theta
^{(j)}=e^{(j)}+\varepsilon _{2}u_{j}\varpi \in \mathbb{R}^{d},$ where $%
e^{(j)}\in \mathbb{R}^{d}$ is the $j$-th unit vector; thus $\left\Vert
\theta ^{(j)}\right\Vert \leq 1+\varepsilon _{2}d\leq 1+d$ (since $%
\varepsilon _{2}\leq \varepsilon \leq 1).$ We have%
\[
(\theta ^{(j)})^{\prime }x_{t}=(\theta ^{(j)})^{\prime
}F(c_{t})=c_{t,j}+\varepsilon _{2}\left( \varpi ^{\prime }F(c_{t})\right)
u_{j} 
\]%
(since $f_{j}(c)=c_{j}$ for $j\leq m),$ and hence%
\begin{eqnarray*}
(a_{t,j}-c_{t,j})^{2}-\psi _{t}^{(j)}(\theta ^{(j)})
&=&(a_{t,j}-c_{t,j})^{2}-(a_{t,j}-c_{t,j}-\varepsilon _{2}\left( \varpi
^{\prime }F(c_{t})\right) u_{j})^{2} \\
&=&2\varepsilon _{2}\left( \varpi ^{\prime }F(c_{t})\right)
u_{j}(a_{t,j}-c_{t,j})-(\varepsilon _{2}\varpi ^{\prime
}F(c_{t}))^{2}u_{j}^{2}.
\end{eqnarray*}%
Summing over $j=1,...,m$ yields%
\begin{eqnarray*}
\sum_{j=1}^{m}\left[ (a_{t,j}-c_{t,j})^{2}-\psi _{t}^{(j)}(\theta ^{(j)})%
\right] &=&2\varepsilon _{2}\left( \varpi ^{\prime }F(c_{t})\right)
u^{\prime }(a_{t}-c_{t})-(\varepsilon _{2}\varpi ^{\prime
}F(c_{t})^{2}\left\Vert u\right\Vert ^{2} \\
&\geq &2\varepsilon _{2}\left( \varpi ^{\prime }F(c_{t})\right) u^{\prime
}(a_{t}-c_{t})-(\varepsilon _{2})^{2}d^{2} \\
&\geq &2\varepsilon _{2}w(c_{t})u^{\prime }(a_{t}-c_{t})-\varepsilon
_{1}\cdot 2\varepsilon _{2}\left\Vert u\right\Vert \left\Vert
a_{t}-c_{t}\right\Vert -(\varepsilon _{2})^{2}d^{2} \\
&\geq &2\varepsilon _{2}w(c_{t})u^{\prime }(a_{t}-c_{t})-2\varepsilon
_{1}\varepsilon _{2}\sqrt{m}-(\varepsilon _{2})^{2}d^{2}
\end{eqnarray*}%
(since: $\left\Vert u\right\Vert =1,$ $|\varpi ^{\prime }F(c)|\leq d$ [the
coordinates of $\varpi $ are between $-1$ and $1$ and those of $F(c)$
between $0$ and $1$],\ $\left\Vert a_{t}-c_{t}\right\Vert \leq \sqrt{m}$
(since $a_{t},c_{t}\in \lbrack 0,1]^{m},$ and recall (\ref{eq:omega-w})).

Together with (\ref{eq:all-j}) we get (recall that $\varepsilon
_{3}=(\varepsilon _{2})^{2}$ and $\varepsilon _{1}=\varepsilon /(2\sqrt{m})$%
):%
\[
2\varepsilon _{2}\cdot \frac{1}{T}\sum_{t=1}^{T}w(c_{t})u^{\prime
}(a_{t}-c_{t})\leq (\varepsilon _{2})^{2}(m+m(1+d)^{2}+d^{2})+\varepsilon
\varepsilon _{2}; 
\]%
hence, dividing by $2\varepsilon _{2}$ and recalling that $\varepsilon
_{2}=\varepsilon /(m+m(1+d)^{2}+d^{2})$):%
\[
u\cdot \frac{1}{T}\sum_{t=1}^{T}w(c_{t})(a_{t}-c_{t})\leq \frac{\varepsilon 
}{2}+\frac{\varepsilon }{2}=\varepsilon . 
\]%
Since $u\in \mathbb{R}^{m}$ with $\left\Vert u\right\Vert =1$ was arbitrary,
the proof of (\ref{eq:(L,eps)-wc}) is complete.

\emph{(v) Grid. }For the \textquotedblleft moreover" statement, let $%
\varepsilon _{4}:=\varepsilon /(L\sqrt{m}+1),$ and take $D\subseteq C$ to be
a finite $\varepsilon _{4}$-grid in $C$; i.e., for every $c\in C$ there is $%
d(c)\in D$ with $||d(c)-c||\leq \varepsilon _{4}.$ Replace the forecast $%
c_{T}$ obtained above with $\tilde{c}_{T}:=d(c_{T});$ then, for every $%
a_{T}\in A,$ we have%
\[
||w(c_{T})(a_{T}-c_{T})-w(\tilde{c}_{T})(a_{T}-\tilde{c}_{T})||\leq
L\varepsilon _{4}\sqrt{m}+\varepsilon _{4}=\varepsilon . 
\]%
Therefore the score $S_{T}^{w}$ changes by at most $\varepsilon $, and so it
is at most $2\varepsilon .$
\end{proof}

\section{Smooth Calibration\label{s:smooth calibration}}

In this section we show (Propositions \ref{p:weak->smooth} and \ref%
{p:smooth->weak}) that weak calibration and smooth calibration are
essentially equivalent (albeit with different constants $\varepsilon ,L$).
The existence of weakly calibrated procedures (Theorem \ref{th:weak-calib},
proved in the previous section) then implies the existence of smoothly
calibrated procedures, which proves Theorem \ref{th:smooth-c}.

\bigskip

We first show how to go from weak to smooth calibration. When comparing the
two scores, we see that the smooth calibration score uses weighted averages
rather than sums: $\sum_{s=1}^{T}\Lambda (c_{s},c_{t})(a_{s}-c_{s})$ is
divided by $\sum_{s=1}^{T}\Lambda (c_{s},c_{t}).$ The following useful lemma
shows how to bound the latter using the former.

\begin{lemma}
\label{l:kappa2K}There exists a constant $0<\gamma \equiv \gamma _{C}<\infty 
$ that depends only on the dimension $m$ and the diameter $\alpha $ of $C$
such that for any $L$-smoothing weight function $\Lambda ,$ any $%
c_{1},...,c_{T}\in C,$ and any\footnote{%
The set $C-C$ consists of all $b=b^{\prime }-b^{\prime \prime }$ with $%
b^{\prime },b^{\prime \prime }\in C.$} $b_{1},...,b_{T}\in C-C,$ putting%
\[
B_{t}:=\sum_{s=1}^{T}\Lambda (c_{s},c_{t})b_{s}\text{\ \ \ and\ \ \ }%
W_{t}:=\sum_{s=1}^{T}\Lambda (c_{s},c_{t}) 
\]%
for all $1\leq t\leq T,$ we have%
\[
\frac{1}{T}\sum_{t=1}^{T}\frac{||B_{t}^{{}}||}{W_{t}^{{}}}\leq \gamma
L^{m/2}\left( \frac{1}{T}\max_{1\leq t\leq T}||B_{t}^{{}}||\right) ^{1/2}. 
\]
\end{lemma}

\begin{proof}
Let $D\subset \mathbb{R}^{m}$ be an $m$-dimensional cube with sides of
length $\alpha $ that contains $C$ (such a cube exists because the diameter
of $C$ is $\alpha ).$ Let $D_{1},...,D_{M}$ be a partition of $D$ into
disjoint cubes with sides of length $1/(2L\sqrt{m});$ the diameter of each
such cube is thus $1/(2L)$, and the number of cubes is $M=\left\lceil
2\alpha L\sqrt{m}\right\rceil ^{m}.$

Put $\kappa :=(1/T)\max_{1\leq t\leq T}||B_{t}||$ and $\bar{b}_{t}\equiv 
\bar{b}_{t}^{\Lambda }:=B_{t}/W_{t}.$ When $W_{t}$ is large, the inequality $%
||B_{t}||/W_{t}\leq \kappa T/W_{t}$ provides a good bound; we will show
that, for a large proportion of indices $t,$ this is indeed the case.

Given $\eta >0$ (which will be specified later), call a cube $D_{j}$ \emph{%
good} if it contains at least $\eta T$ elements of the sequence $%
c_{1},...,c_{T}$ (i.e., $|\{t\leq T:c_{t}\in D_{j}\}|\geq \eta T),$ and $%
\emph{bad}$ otherwise; call an index $t\leq T$ \emph{good }if\emph{\ }$c_{t}$
belongs to some good cube $D_{i},$ and \emph{bad} otherwise.

If $c_{t}$ and $c_{s}$ belong to the same cube $D_{j}$ then $%
||c_{s}-c_{t}||\leq \mathrm{diam}(D_{j})=1/(2L),$ and so $1-\Lambda
(c_{s},c_{t})=\Lambda (c_{t},c_{t})-\Lambda (c_{s},c_{t})\leq
L||c_{s}-c_{t}||\leq 1/2,$ which implies that $\Lambda (c_{s},c_{t})\geq
1/2. $ Therefore for every good $t$ we have $W_{t}=\sum_{s\leq T}\Lambda
(c_{s},c_{t})\geq (1/2)\eta T,$ and thus $\left\Vert \bar{b}_{t}\right\Vert
\leq 2T\kappa /(\eta T)=2\kappa /\eta ,$ which then gives 
\[
\frac{1}{T}\sum_{good\text{ }t\leq T}||\bar{b}_{t}||\leq \frac{2\kappa }{%
\eta }. 
\]

The number of bad $t$ is less than $M\cdot \eta T$ (because each bad cube
contains less than $\eta T$ elements of $c_{1},...,c_{T},$ and there are $M$
cubes). For every $s$ we have $||b_{s}||\leq \mathrm{diam}(C)=\alpha ,$ and
so $\bar{b}_{t},$ which is a weighted average of the $b_{s},$ satisfies $||%
\bar{b}_{t}||\leq \alpha $ as well. Thus%
\[
\frac{1}{T}\sum_{bad\text{ }t\leq T}||\bar{b}_{t}||\leq \frac{1}{T}M\cdot
\eta T\cdot \alpha =\alpha \eta M. 
\]

Adding the last two displayed inequalities and choosing $\eta =\sqrt{%
(2\kappa )/(\alpha M)}$ yields 
\[
\frac{1}{T}\sum_{t\leq T}||\bar{b}_{t}||\leq 2\sqrt{2\kappa \alpha M}; 
\]%
recalling that $M=\left\lceil 2\alpha L\sqrt{m}\right\rceil ^{m}$ gives the
result, with $\gamma $ essentially equal to $2^{(m+3)/2}\alpha
^{m/2}m^{m/4}. $
\end{proof}

\bigskip

An immediate consequence is that smooth calibration is a weaker requirement
than calibration.

\begin{corollary}
\label{c:K-to-K-lambda}Calibration implies smooth calibration.
\end{corollary}

\begin{proof}
We will show that\footnote{%
The notations $f(x)=\mathrm{O}(g(x))$, $f(x)=\Omega (g(x)),$ and $%
f(x)=\Theta (g(x))$ mean, as usual, that there are constants $C<\infty $ and 
$c>0$ such that for all $x$ we have, respectively, $f(x)\leq Cg(x),$ $%
f(x)\geq cg(x),$ and $cg(x)\leq f(x)\leq Cg(x).$ In our case $x$ stands for $%
(\varepsilon ,L)$; the dimension $m$ is assumed fixed.} $K_{T}^{\Lambda }=%
\mathrm{O}\left( \sqrt{K_{T}}\right) $ for each fixed $L$ (where $%
K_{T}^{\Lambda }$ is the $\Lambda $-smoothly calibrated score for any $L$%
-Lipschitz $\Lambda ,$ and $K_{T}$ is the regular calibration score).
Indeed, for every $c_{t}\in C$ we have (use $0\leq \Lambda (\cdot ,\cdot
)\leq 1)$ 
\begin{eqnarray*}
\frac{1}{T}\left\Vert \sum_{s=1}^{T}\Lambda
(c_{s},c_{t})(a_{s}-c_{s})\right\Vert &\leq &\frac{1}{T}\sum_{s=1}^{T}%
\Lambda (c_{s},c_{t})||(a_{s}-c_{s})|| \\
&\leq &\frac{1}{T}\sum_{s=1}^{T}||a_{s}-c_{s}||=K_{T}.
\end{eqnarray*}%
Now apply Lemma \ref{l:kappa2K} with $b_{s}=a_{s}-c_{s}$ for all $s.$
\end{proof}

\bigskip

Returning to our proof, we can finally obtain smoothly calibrated procedures
with the desired properties, proving Theorem\ref{th:smooth-c}.

\begin{proposition}
\label{p:weak->smooth}An $(\varepsilon ,L)$-weakly calibrated procedure is $%
(\varepsilon ^{\prime },L)$-smoothly calibrated for $\varepsilon ^{\prime
}=\gamma L^{m/2}\varepsilon ^{1/2}$ (with the constant $\gamma \equiv \gamma
_{C}$ given by Lemma \ref{l:kappa2K}).
\end{proposition}

\begin{proof}
For any $L$-Lipschitz smoothing function $\Lambda ,$ Lemma \ref{l:kappa2K}
with $b_{s}=a_{s}-c_{s}$ for all $s$ yields%
\[
K_{T}^{\Lambda }\leq \gamma L^{m/2}\left( \sup_{w\in W_{L}}S_{T}^{w}\right)
^{1/2} 
\]%
(because $\Lambda (\cdot ,c_{t})\in W_{L}$ for all $c_{t}$; recall the
definition (\ref{eq:(L,eps)-wc}) of $S_{T}^{w}$ in Section \ref{s:weak
calibration}). Therefore $\sup_{w\in W_{L}}S_{T}^{w}\leq \varepsilon $
implies $K_{T}^{\Lambda }\leq \gamma L^{m/2}\varepsilon ^{1/2}=\varepsilon
^{\prime }$
\end{proof}

\bigskip

\begin{proof}[Proof of Theorem \protect\ref{th:smooth-c}]
Apply Theorem \ref{th:weak-calib} and Proposition \ref{p:weak->smooth}, and
recall (Section \ref{sus:leaky}) that for deterministic procedures leaks do
not matter.
\end{proof}

\bigskip

As an aside, we now show how to go from smooth to weak calibration.

\begin{proposition}
\label{p:smooth->weak}An $(\varepsilon ,L)$-smoothly calibrated procedure is 
$(\varepsilon ^{\prime },L^{\prime })$-weakly calibrated, where\footnote{%
We have not tried to optimize the estimates for $\varepsilon ^{\prime }$ and 
$L^{\prime }.$} $\varepsilon ^{\prime }=\Omega \left( \varepsilon
^{1/2}L^{m/2}\right) $ and $L^{\prime }=\mathrm{O}\left( \varepsilon
^{1/2}L^{(m+2)/2}\right) .$
\end{proposition}

\begin{proof}
Let $(D_{j})_{j=1,...,M}$ be a partition of $[0,1]^{m}\supseteq C$ into
disjoint cubes with sides $\delta :=1/(L\sqrt{m});$ the diameter of each
cube is thus $\delta \sqrt{m}=1/L$, and the number of cubes is $M=\delta
^{-m}=L^{m}m^{m/2}.$ Let $\varepsilon _{1}:=\sqrt{\varepsilon L^{m}}$ and $%
\varepsilon _{2}:=\varepsilon _{1}m^{m/2}/M=\sqrt{\varepsilon /L^{m}}.$ Take 
$L^{\prime }:=L\varepsilon _{1}/2=\sqrt{\varepsilon L^{m+2}}/2$ and $%
\varepsilon ^{\prime }=\varepsilon _{1}(1+\sqrt{m}+\sqrt{m^{m+1}})=\sqrt{%
\varepsilon L^{m}}(1+\sqrt{m}+\sqrt{m^{m+1}}).$

Fix $a_{t},c_{t}\in C\subseteq \lbrack 0,1]^{m}$ for $t=1,...,T,$ and a
weight function $w$ in $W_{L^{\prime }}.$ Assume that $K_{T}^{\Lambda }\leq
\varepsilon $ holds for every smoothing function $\Lambda $ that is $L$%
-Lipschitz in the first coordinate; we will show that $S_{T}^{w}\leq
\varepsilon ^{\prime }$ (where $K_{T}^{\Lambda }$ and $S_{T}^{w}$ are given
by (\ref{eq:K-def}) and (\ref{eq:(L,eps)-wc}), respectively).

Let $V\subseteq \{1,...,T\}$ be the set of indices $t$ such that the cube $%
D_{j}$ that contains $c_{t}$ includes at least a fraction $\varepsilon _{2}$
of $c_{1},...,c_{T},$ i.e., $|\{s\leq T:c_{s}\in D_{j}\}|\geq \varepsilon
_{2}T.$ Then 
\begin{equation}
T-|V|=|\{t\leq T:t\notin V\}|<M\cdot \varepsilon _{2}T=\varepsilon _{2}MT,
\label{eq:not-V}
\end{equation}%
because there are at most $M$ cubes containing less than $\varepsilon _{2}T$
points each.

We distinguish two cases.

\emph{Case 1}: $\max_{t\in V}w(c_{t})<\varepsilon _{1}.$ Since $%
||a_{t}-c_{t}||\leq \sqrt{m},$ we have $||\sum_{t\in
V}w(c_{t})(a_{t}-c_{t})||\leq |V|\cdot \varepsilon _{1}\cdot \sqrt{m}\leq
\varepsilon _{1}\sqrt{m}T$ (use $|V|\leq T),$ and $||\sum_{t\notin
V}w(c_{t})(a_{t}-c_{t})||\leq (T-|V|)\cdot 1\cdot \sqrt{m}=\varepsilon _{2}M%
\sqrt{m}T$ (use (\ref{eq:not-V})). Adding and dividing by $T$ yields 
\[
S_{T}^{w}\leq (\varepsilon _{1}+\varepsilon _{2}M)\sqrt{m}=\varepsilon _{1}%
\sqrt{m}(1+m^{m/2})<K\varepsilon _{1}=\varepsilon ^{\prime }. 
\]

\emph{Case 2}: $\max_{t\in V}w(c_{t})\geq \varepsilon _{1}.$ Let $s\in V$ be
such that $w(c_{s})=\max_{t\in V}w(c_{t})\geq \varepsilon _{1},$ and let $%
R\subseteq V$ be the set of indices $r$ such that $c_{r}$ lies in the same
cube $D_{j}$ as $c_{s}.$

For each $r$ in $R,$ proceed as follows. First, we have $|w(c_{s})-w(c_{r})|%
\leq L^{\prime }||c_{s}-c_{r}||\leq L^{\prime }\cdot \delta \sqrt{m}%
=L\varepsilon _{1}/2\cdot (1/L)=\varepsilon _{1}/2,$ and so 
\begin{equation}
w(c_{r})\geq w(c_{s})-\frac{\varepsilon _{1}}{2}\geq \varepsilon _{1}-\frac{%
\varepsilon _{1}}{2}=\frac{\varepsilon _{1}}{2}.  \label{eq:w(cr)}
\end{equation}%
Next, put $w^{r}(c):=\min \{w(c),w(c_{r})\}$ and $\Lambda
(c,c_{r}):=w^{r}(c)/w(c_{r})$ for $r\in R$ (and, for $t\notin R,$ put, say, $%
\Lambda (c,c_{t})=1$ for all $c)$; then $\mathcal{L}(\Lambda (\cdot
,c_{r}))\leq \mathcal{L}(w)/w(c_{r})\leq L^{\prime }/(\varepsilon _{1}/2)=L,$
and so, by our assumption 
\begin{equation}
\frac{1}{T}\sum_{r\in R}||\bar{a}_{r}^{\Lambda }-c_{r}^{\Lambda }||\leq 
\frac{1}{T}\sum_{t\leq T}||\bar{a}_{t}^{\Lambda }-c_{t}^{\Lambda }||\leq
K_{T}^{\Lambda }\leq \varepsilon .  \label{eq:TK}
\end{equation}

We will now show that $S_{T}^{w}$ is close to an appropriate multiple of $||%
\bar{a}_{r}^{\Lambda }-c_{r}^{\Lambda }||,$ for each $r$ in $R.$ For $t$ in $%
V$ we have $w(c_{t})\leq w(c_{s}),$ and so $0\leq w(c_{t})-w^{r}(c_{t})\leq
w(c_{s})-w(c_{r})\leq \varepsilon _{1}/2$ (recall (\ref{eq:w(cr)})), which
gives%
\[
\left\Vert \sum_{t\in V}\left( w(c_{t})-w^{r}(c_{t})\right)
(a_{t}-c_{t})\right\Vert \leq |V|\cdot \frac{\varepsilon _{1}}{2}\cdot \sqrt{%
m}\leq \frac{1}{2}\varepsilon _{1}\sqrt{m}T. 
\]%
For $t\notin V$ we have $0\leq w(c_{t})-w^{r}(c_{t})\leq 1,$ and so (recall (%
\ref{eq:not-V})) 
\[
\left\Vert \sum_{t\notin V}\left( w(c_{t})-w^{r}(c_{t})\right)
(a_{t}-c_{t})\right\Vert \leq (T-|V|)\cdot 1\cdot \sqrt{m}\leq \varepsilon
_{2}M\sqrt{m}T. 
\]%
Adding the two inequalities and dividing by $T$ yields%
\begin{eqnarray*}
S_{T}^{w} &=&\left\Vert \frac{1}{T}\sum_{t=1}^{T}w(c_{t})(a_{t}-c_{t})\right%
\Vert \leq \left\Vert \frac{1}{T}\sum_{t=1}^{T}w^{r}(c_{t})(a_{t}-c_{t})%
\right\Vert +\sqrt{m}\left( \frac{\varepsilon _{1}}{2}+\varepsilon
_{2}M\right) \\
&\leq &||\bar{a}_{r}^{\Lambda }-c_{r}^{\Lambda }||+\sqrt{m}\left( \frac{%
\varepsilon _{1}}{2}+\varepsilon _{2}M\right) ,
\end{eqnarray*}%
because $\sum_{t\leq T}w^{r}(c_{t})(a_{t}-c_{t})=\left( \sum_{t\leq
T}w^{r}(c_{t})\right) (\bar{a}_{r}^{\Lambda }-c_{r}^{\Lambda })$ and $%
\sum_{t\leq T}w^{r}(c_{t})\leq T.$ The set $R$ contains at least $%
\varepsilon _{2}T$ points (these are all the points in the same cube as $%
c_{s}),$ i.e., $|R|\geq \varepsilon _{2}T;$ averaging over all $r$ in the
set $R$ and then recalling (\ref{eq:TK}) finally gives 
\begin{eqnarray*}
S_{T}^{w} &\leq &\frac{1}{|R|}\sum_{r\in R}||\bar{a}_{r}^{\Lambda
}-c_{r}^{\Lambda }||+\sqrt{m}\left( \frac{\varepsilon _{1}}{2}+\varepsilon
_{2}M\right) \\
&\leq &\frac{1}{\varepsilon _{2}}\frac{1}{T}\sum_{r\in R}||\bar{a}%
_{r}^{\Lambda }-c_{r}^{\Lambda }||+\sqrt{m}\left( \frac{\varepsilon _{1}}{2}%
+\varepsilon _{2}M\right) \\
&\leq &\frac{1}{\varepsilon _{2}}\varepsilon +\sqrt{m}\left( \frac{%
\varepsilon _{1}}{2}+\varepsilon _{2}M\right) =\frac{\varepsilon \sqrt{L^{m}}%
}{\sqrt{\varepsilon }}+\varepsilon _{1}\left( \frac{\sqrt{m}}{2}%
+m^{(m+1)/2}\right) \\
&=&\varepsilon _{1}\left( 1+\frac{\sqrt{m}}{2}+m^{(m+1)/2}\right)
<K\varepsilon _{1}=\varepsilon ^{\prime },
\end{eqnarray*}%
completing the proof.
\end{proof}

\section{Nash Equilibrium Dynamics\label{s:Nash}}

In this section we use our results on smooth calibration to obtain dynamics
in $n$-person games that are in the long run close to Nash equilibria most
of the time.

A (finite) \emph{game} is given by a finite set of players $N$, and, for
each player $i\in N,$ a finite set of actions\footnote{%
We refer to one-shot choices as \textquotedblleft actions" rather than
\textquotedblleft strategies," the latter term being reserved for repeated
interactions.} $A^{i}$ and a payoff function $u^{i}:A\rightarrow \mathbb{R},$
where $A:=\prod_{i\in N}A^{i}$ denotes the set of action combinations of all
players. Let $n:=|N|$ be the number of players, $m^{i}:=|A^{i}|$ the number
of pure actions of player $i,$ and $m:=\sum_{i\in N}m^{i}$; also, let $U$ be
a bound on payoffs, i.e., $|u^{i}(a)|\leq U$ for all $a\in A$ and $i\in N.$
The set of mixed actions of player $i$ is $X^{i}:=\Delta (A^{i}),$ the unit
simplex (i.e., the set of probability distributions) on $A^{i}$; we identify
the pure actions in $A^{i}$ with the unit vectors of $X^{i},$ and so $%
A^{i}\subseteq X^{i}.$ Put $C\equiv X:=\prod_{i\in N}X^{i}$ for the set of
mixed-action combinations$.$ The payoff functions $u^{i}$ are multilinearly
extended to $X,$ and thus $u^{i}:X\rightarrow \mathbb{R}.$

For each player $i$, a combination of mixed actions of the other players $%
x^{-i}=(x^{j})_{j\neq i}\in \prod_{j\neq i}X^{j}=:X^{-i},$ and $\varepsilon
\geq 0,$ let $\mathrm{BR}_{\varepsilon }^{i}(x^{-i}):=\{x^{i}\in
X^{i}:u^{i}(x^{i},x^{-i})\geq \max_{y^{i}\in
X^{i}}u^{i}(y^{i},x^{-i})-\varepsilon \}$ denote the set of $\varepsilon $%
\emph{-best replies} of $i$ to $x^{-i}.$ A (mixed) action combination $x\in
X $ is a \emph{Nash }$\varepsilon $\emph{-equilibrium} if $x^{i}\in \mathrm{%
BR}_{\varepsilon }^{i}(x^{-i})$ for every $i\in N;$ let \textrm{NE}$%
(\varepsilon )\subseteq X$ denote the set of Nash $\varepsilon $-equilibria
of the game.

A (discrete-time) \emph{dynamic }consists of each player $i\in N$ playing a
pure action $a_{t}^{i}\in A^{i}$ at each time period $t=1,2,...;$ put $%
a_{t}=(a_{t}^{i})_{i\in N}\in A.$ There is perfect monitoring: at the end of
period $t$ all players observe $a_{t}$. The dynamic is \emph{uncoupled }%
(Hart and Mas-Colell 2003, 2006, 2013) if the play of every player $i$ may
depend only on player $i$'s payoff function $u^{i}$ (and not on the other
players' payoff functions). Formally, such a dynamic is given by a mapping
for each player $i$ from the history $h_{t-1}=(a_{1},...,a_{t-1})$ and his
own payoff function $u^{i}$ into $X^{i}=\Delta (A^{i})$ (player $i$'s choice
may be random); we will call such mappings \emph{uncoupled.} Let $%
x_{t}^{i}\in X^{i}$ denote the mixed action that player $i$ plays at time $%
t, $ and put $x_{t}=(x_{t}^{i})_{i\in N}\in X.$

The dynamics we consider are smooth variants of the \textquotedblleft
calibrated learning" introduced by Foster and Vohra (1997). \emph{Calibrated
learning} consists of each player best-replying to calibrated forecasts on
the other players' actions; it results in the joint distribution of play
converging in the long run to the set of correlated equilibria of the game.
Kakade and Foster (2004) defined \emph{publicly calibrated learning}, where
each player approximately best-replies to a public weakly calibrated
forecast on the joint actions of all players, and proved that most of the
time the play is an approximate Nash equilibrium. We consider instead \emph{%
smooth calibrated learning}, where weak calibration is replaced with the
more natural smooth calibration; it amounts to taking calibrated learning
and smoothing out both the forecasts and the best replies. Moreover, our
forecasts are $n$-tuples of mixed strategies (in $\prod_{i}\Delta (A^{i})$),
rather than correlated mixtures (in $\Delta (\prod_{i}A^{i})$).

Formally, a \emph{smooth calibrated learning} dynamic is given by:

\begin{enumerate}
\item[(D1)] An $(\varepsilon _{c},L_{c})$-smoothly calibrated deterministic
procedure, which yields at time $t$ a forecast $c_{t}\in X$ on the
distribution of actions of each player$.$

\item[(D2)] For each player $i\in N$ an $L_{g}$-Lipschitz $\varepsilon _{g}$%
-approximate best-reply mapping $g^{i}:X\rightarrow X^{i};$ i.e., $%
g^{i}(x)\in \mathrm{BR}_{\varepsilon _{g}}^{i}(x^{-i})$ for every $x^{-i}\in
X^{-i}.$

\item[(D3)] Each player runs the procedure in (D1), generating at time $t$ a
forecast $c_{t}\in X;$ then each player $i$ plays at period $t$ the mixed
action\footnote{%
Thus $\mathbb{P}\left[ a_{t}=a~|~h_{t-1}\right] =\prod_{i\in
N}x_{t}^{i}(a^{i})$ for every $a=(a^{i})_{i\in N}\in A,$ where $h_{t-1}$ is
the history and $x_{t}^{i}(a^{i})$ is the probability that $x_{t}^{i}\in
\Delta (A^{i})$ assigns to the pure action $a^{i}\in A^{i}.$} $%
x_{t}^{i}:=g^{i}(c_{t})\in X^{i},$ where $g^{i}$ is given by (D2). All
players observe the action combination $a_{t}=(a_{t}^{i})_{i\in N}\in A$
that has actually been played, and remember it.
\end{enumerate}

The existence of a deterministic smoothly calibrated procedure in (D1) is
given by Theorem \ref{th:smooth-c}. For each player $i,$ the payoff function 
$u^{i}$ is linear in $x^{i},$ and $|u^{i}(x^{i},x^{-i})-u^{i}(y^{i},x^{-i})|%
\leq \sqrt{m^{i}}U||x^{i}-y^{i}||,$ and so $\mathcal{L}(u^{i})\leq \sqrt{%
m^{i}}U\leq \sqrt{m}U;$ the existence of Lipschitz approximate best-reply
mappings in (D2) is then given by Lemma \ref{p:BR} in the Appendix (in
particular, for $\varepsilon _{g}$ and $L_{g}$ such that $L_{g}\geq \nu _{m}(%
\sqrt{m}U/\varepsilon _{g})^{m+1}).$

Since for each player $i$ the approximate best reply condition in (D2) makes
use \emph{only} of player $i$'s payoff function $u^{i},$ we can without loss
of generality choose $g^{i}$ so as to depend only on $u^{i},$ which makes
the dynamic \emph{uncoupled} (see above).

Our result is:

\begin{theorem}
\label{th:finite}Fix the finite set of players $N,$ the finite action spaces 
$A^{i}$ for all $i\in N,$ and the payoff bound $U<\infty .$ For every $%
\varepsilon >0,$ any smooth calibrated learning dynamic with appropriate
parameters\footnote{%
Such as those given in (\ref{eq:parameters}).} is an uncoupled dynamic that
satisfies 
\[
\liminf%
_{T\rightarrow \infty }\frac{1}{T}\left\vert \{t\leq T:x_{t}\in \mathrm{NE}%
(\varepsilon )\}\right\vert \geq 1-\varepsilon \text{\hspace{0.2in}}\mathrm{%
(a.s.)} 
\]%
for every finite game with payoff functions $(u^{i})_{i\in N}$ that are
bounded by $U$ (i.e., $|u^{i}(a)|\leq U$ for all $i\in N$ and $a\in A).$
\end{theorem}

\bigskip

The idea of the proof is as follows. First, assume that the forecasts $c_{t}$
are in fact calibrated (rather than just smoothly calibrated) and, moreover,
that they are calibrated with respect to the mixed plays $x_{t}$ (rather
than with respect to the actual plays $a_{t}).$ Because $x_{t}$ is given by
a fixed function of $c_{t},$ namely, $x_{t}=g(c_{t})\equiv
(g^{i}(c_{t}))_{i\in N},$ the sequence of mixed plays in those periods when
the forecast was a certain $c$ is the constant sequence $g(c),...,g(c),$
whose average is $g(c),$ and calibration then implies that $g(c)$ must be
close to $c$ (most of the time, i.e., for forecasts that appear with
positive frequency). But we have only smooth calibration; however, because $%
g $ is a continuous function, if $c$ and $g(c)$ are far from one another
then so are $c^{\prime }$ and $g(c^{\prime })$ for any $c^{\prime }$ close
to $c,$ and so the average of such $g(c^{\prime })$ is also far from $c,$
contradicting smooth calibration. Thus, most of the time $g(c_{t})$ is close
to $c_{t},$ and hence $g(g(c_{t}))$ is close to $g(c_{t})$ (because $g$ is
continuous)---which says that $g(c_{t})$ is close to an approximate best
reply to itself, i.e., $g(c_{t})$ is an approximate Nash equilibrium.
Finally, an appropriate use of a strong law of large numbers shows that if
the actual plays $a_{t}$ are (smoothly) calibrated then so are their
expectations, i.e., the mixed plays $x_{t}.$ Two crucial features of our
dynamic---which are needed to get Nash equilibria, and cannot be obtained
with standard, probabilistic, calibration---are, first, that all players
always have the same forecast, and second, that (smooth) calibration is
preserved despite the fact that the actions depend on the forecasts
(leakyness).

\bigskip

\begin{proof}
This proof goes along similar lines to the proof of Kakade and Foster (2004)
for publicly calibrated dynamics (which is the only other calibration-based
Nash dynamic to date\footnote{%
Recall footnote \ref{ftn:fixed-pt-calibration}.}).

Recall that $m^{i}:=|A^{i}|$ and $m:=\sum_{i\in N}m^{i},$ and so $X\subset
\lbrack 0,1]^{m}.$ Put $g(c):=(g^{i}(c))_{i\in N}$ for every $c\in X;$ thus $%
g:X\rightarrow X$ is a Lipschitz function with $\mathcal{L}(g)\leq nL_{g}$
(because $\mathcal{L}(g^{i})\leq L_{g}$ for each $i).$

Take $\Lambda $ to be the $L_{c}$-tent smoothing function: $\Lambda
(c^{\prime },c)=[1-L_{c}||c^{\prime }-c||]_{+}$ for all $c,c^{\prime }\in X.$

For each period $t,$ let $c_{t}\in X$ be the \emph{forecast}, $%
x_{t}=g(c_{t})\in X$ the \emph{behavior} (i.e., mixed actions), and $%
a_{t}\in A$ the realized pure \emph{actions} ($c_{t},x_{t},$ and $a_{t}$ all
depend on the history). We divide the proof into the following steps: (i)
smoothed average actions $\bar{a}_{t}^{\Lambda }$ and forecasts $\bar{c}%
_{t}^{\Lambda }$ are close (by smooth calibration); (ii) smoothed average
actions $\bar{a}_{t}^{\Lambda }$ and behaviors $\bar{x}_{t}^{\Lambda }$ are
close (by the law of large numbers); (iii) forecasts $c_{t}$ and behaviors $%
x_{t}$ are close (because smoothing had little effect there); (iv) behaviors 
$x_{t}$ are close to Nash equilibria. Finally, (v) shows how to tweak the
parameters to get the desired result.

\emph{(i) Smoothed average actions and smoothed forecasts are close.}

Let $T_{0}$ be such that the smooth calibration score $K_{T}^{\Lambda }\leq
\varepsilon _{c}$ for all $T>T_{0},$ i.e.,%
\begin{equation}
\frac{1}{T}\sum_{t=1}^{T}||\bar{a}_{t}^{\Lambda }-c_{t}^{\Lambda }||\leq
\varepsilon _{c}  \label{eq:a-c}
\end{equation}%
for all $T>T_{0}.$

\emph{(ii) Smoothed average actions and smoothed average behaviors are close.%
}

Let $D\subset X$ be a finite $\varepsilon _{1}$-grid of $X.$ For each $d\in
D $ we have $\mathbb{E}\left[ \Lambda (c_{s},d)a_{s}~|~h_{s-1}\right]
=\Lambda (c_{s},d)x_{s}$ (given $h_{s-1},$ only $a_{s}$ is random, and its
conditional expectation is $\mathbb{E}\left[ a_{s}~|~h_{s-1}\right]
=g(c_{s})=x_{s}$). By the Strong Law of Large Numbers for Dependent Random
Variables (see Lo\`{e}ve 1978, Theorem 32.1.E: $(1/T)\sum_{s=1}^{T}(X_{s}-%
\mathbb{E}\left[ X_{s}|h_{s-1}\right] )\rightarrow 0$ as $T\rightarrow
\infty $ a.s., for random variables $X_{s}$ that are, in particular,
uniformly bounded; note that there are finitely many $d\in D$) we get%
\begin{equation}
\lim_{T\rightarrow \infty }\frac{1}{T}\sum_{s=1}^{T}\Lambda
(c_{s},d)(a_{s}-x_{s})=0\text{\ \ for all }d\in D\text{\hspace{0.2in}\textrm{%
(a.s.)}}.  \label{eq:SLLN}
\end{equation}%
Thus, for each one of the (almost all) infinite histories $h_{\infty }$
where (\ref{eq:SLLN}) holds, there is a finite $T_{1}\equiv T_{1}(h_{\infty
})$ such that $(1/T)\left\Vert \sum_{t=1}^{T}\Lambda
(c_{s},d)(a_{s}-x_{s})\right\Vert \leq \varepsilon _{1}$ for all $T>T_{1}$
and all $d\in D.$ Now for every $c\in X$ there is $d\in D$ with $||d-x||\leq
\varepsilon _{1},$ and so $|\Lambda (c_{s},d)-\Lambda (c_{s},c)|\leq
L_{c}||c-d||\leq L_{c}\varepsilon _{1};$ together with\footnote{%
Because $a_{s},x_{s}\in \lbrack 0,1]^{m}$.} $||a_{s}-x_{s}||\leq \sqrt{m}$
it follows that%
\[
\frac{1}{T}\left\Vert \sum_{s=1}^{T}\Lambda
(c_{s},c)(a_{s}-x_{s})\right\Vert \leq (1+\sqrt{m}L_{c})\varepsilon _{1}%
\text{\ \ for all }T>T_{1}\text{ and all }c\in X. 
\]%
Taking in particular $c=c_{t}$ for all $t\leq T,$ and then applying Lemma %
\ref{l:kappa2K} to the set $[0,1]^{m}$ and $b_{s}=a_{s}-x_{s},$ yields%
\begin{equation}
\frac{1}{T}\sum_{t=1}^{T}||\bar{a}_{t}^{\Lambda }-\bar{x}_{t}^{\Lambda
}||\leq \gamma _{m}L_{c}^{m/2}\sqrt{(1+\sqrt{m}L_{c})\varepsilon _{1}}%
=:\varepsilon _{2},  \label{eq:a-x}
\end{equation}%
where the constant $\gamma _{m}$ depends only on $m.$

\emph{(iii) Behaviors and forecasts are close.}

Because $\Lambda (c_{s},c_{t})>0$ only when $||c_{s}-c_{t}||<1/L_{c},$ it
follows that $c_{t}^{\Lambda },$ as a weighted average of such $c_{s},$
satisfies $||c_{t}^{\Lambda }-c_{t}||<1/L_{c}.$ Moreover, $%
||x_{s}-x_{t}||=||g(c_{s})-g(c_{t})||\leq nL_{g}/L_{c},$ and so $||\bar{x}%
_{t}^{\Lambda }-x_{t}||\leq nL_{g}/L_{c},$ which together with (\ref{eq:a-c}%
) and (\ref{eq:a-x}) gives%
\begin{equation}
\frac{1}{T}\sum_{t=1}^{T}||x_{t}-c_{t}||\leq \varepsilon _{c}+\varepsilon
_{2}+\frac{1}{L_{c}}+\frac{nL_{g}}{L_{c}}=:\varepsilon _{3}
\label{eq:eps-eps1}
\end{equation}%
for almost every infinite history $h_{\infty }$ and for every $T>\max
\{T_{0},T_{1}(h_{\infty })\}.$

\emph{(iv) Behaviors are close to Nash equilibria. }From (\ref{eq:eps-eps1})
it immediately follows that, for every $\varepsilon _{4}>0,$ 
\begin{equation}
\frac{1}{T}\left\vert \left\{ t\leq T:\left\Vert g(c_{t})-c_{t}\right\Vert
>\varepsilon _{4}\right\} \right\vert \leq \frac{1}{\varepsilon _{4}}\frac{1%
}{T}\sum_{t=1}^{T}\left\Vert (g(c_{t})-c_{t})\right\Vert \leq \frac{%
\varepsilon _{3}}{\varepsilon _{4}}.  \label{eq:eps1}
\end{equation}%
If $\left\Vert g(c_{t})-c_{t}\right\Vert \leq \varepsilon _{4}$ then%
\[
\left\Vert g^{i}(x_{t})-x_{t}^{i}\right\Vert =\left\Vert
g^{i}(g(c_{t}))-g^{i}(c_{t})\right\Vert \leq L_{g}\varepsilon _{4}, 
\]%
and so%
\begin{eqnarray*}
u^{i}(x_{t}) &\geq &u^{i}(g^{i}(x_{t}),x_{t}^{-i})-\sqrt{m^{i}}U\left\Vert
g^{i}(x_{t})-x_{t}^{i}\right\Vert \\
&\geq &\max_{y^{i}\in \Delta (A^{i})}u^{i}(y^{i},x_{t}^{-i})-\varepsilon
_{g}-\sqrt{m^{i}}UL_{g}\varepsilon _{4}
\end{eqnarray*}%
(for the second inequality we have used $g^{i}(x)\in \mathrm{BR}%
_{\varepsilon _{g}}^{i}(x^{-i})$). Therefore $\left\Vert
g(c_{t})-c_{t}\right\Vert \nolinebreak \leq \nolinebreak \varepsilon _{4}$
implies that $x_{t}\in \mathrm{NE}(\varepsilon _{5}),$ where%
\begin{equation}
\varepsilon _{5}:=\varepsilon _{g}+\sqrt{m}UL_{g}\varepsilon _{4}
\label{eq:eps5}
\end{equation}%
(recall that $m=\sum_{i}m^{i}\geq m^{i}),$ and so, from (\ref{eq:eps-eps1})
and (\ref{eq:eps1}) we get%
\[
\frac{1}{T}\left\vert \left\{ t\leq T:x_{t}\notin \mathrm{NE}(\varepsilon
_{5})\right\} \right\vert \leq \frac{\varepsilon _{3}}{\varepsilon _{4}} 
\]%
for all large enough $T$, for almost every infinite history.

\emph{(v) Tweaking the parameters. }To bound both $\varepsilon
_{3}/\varepsilon _{4}$ and $\varepsilon _{5}$ by, say, $3\varepsilon ,$ one
may take, for instance (see (\ref{eq:a-x})--(\ref{eq:eps5}) and recall Lemma %
\ref{p:BR} in the Appendix),%
\begin{eqnarray}
\varepsilon _{g} &=&\varepsilon ,\;\;L_{g}=\nu _{m}\left( \frac{\sqrt{m}U}{%
\varepsilon }\right) ^{m+1},  \label{eq:parameters} \\
\varepsilon _{4} &=&\frac{2\varepsilon }{\sqrt{m}UL_{g}},  \nonumber \\
\varepsilon _{c} &=&\varepsilon \varepsilon _{4},\;\;L_{c}=\frac{1+nL_{g}}{%
\varepsilon \varepsilon _{4}},  \nonumber \\
\varepsilon _{2} &=&\varepsilon \varepsilon _{4},\;\;\varepsilon _{1}=\frac{%
\varepsilon _{2}^{2}}{\gamma _{m}^{2}L_{c}^{m}(1+\sqrt{m}L_{c}^{{}})}, 
\nonumber
\end{eqnarray}%
because we then get $\varepsilon _{5}=\varepsilon +2\varepsilon
=3\varepsilon $ and $\varepsilon _{3}=\varepsilon \varepsilon
_{4}+\varepsilon \varepsilon _{4}+\varepsilon \varepsilon _{4}=3\varepsilon
\varepsilon _{4}.$
\end{proof}

\bigskip

\noindent \textbf{Remarks. }\emph{(a) Nash dynamics. }Uncoupled dynamics
where Nash $\varepsilon $-equilibria are played $1-\varepsilon $ of the time
were first proposed by Foster and Young (2003), followed by Kakade and
Foster (2004), Foster and Young (2006), Hart and Mas-Colell (2006), Germano
and Lugosi (2007), Young (2009), Babichenko (2012), and others (see also
Remark (h)\ below).

\emph{(b) Coordination.} All players need to coordinate before playing the
game on the smoothly calibrated procedure that they will run; thus, at every
period $t$ they all generate the same forecast $t.$ By contrast, in the
original calibrated learning dynamic of Foster and Vohra (1997)---which
leads to correlated equilibria---every player may use his own calibrated
procedure.

This fits the so-called \emph{Conservation Coordination Law} for game
dynamics, which says that some form of \textquotedblleft coordination" must
be present, either in the limit static equilibrium concept (such as
correlated equilibrium) or in the dynamic leading to it (such as Nash
equilibrium dynamics). See Hart and Mas-Colell (2003, footnote 19) and Hart
(2005, footnote 19).

\emph{(c)} \emph{Deterministic calibration.} In order for all the players to
generate the same forecasts, it is not enough that they all use the same
procedure; in addition, the forecasts must be deterministic (otherwise the
randomizations, which are carried out independently by the players, may lead
to different actual forecasts). This is the reason that we use smoothly
calibrated procedures rather than fully calibrated ones (cf. Oakes 1985 and
Foster and Vohra 1998).

\emph{(d) Leaky calibration}. One may use a common \emph{randomized}
smoothly calibrated procedure, provided that the randomizations are carried
out publicly (i.e., they must be leaked!). Alternatively, a
\textquotedblleft central bureau of statistics" may be used each period to
provide the forecast to all the players.

\emph{(e) Finite memory. }In (D1) one may use a smoothly calibrated
procedure that has finite recall and is stationary (see Theorem \ref%
{th:smooth-c}). However, while in the calibration game of Section \ref%
{s:smooth calibration} both the actions $a_{t}$ and the forecasts $c_{t}$
are monitored and thus become part of the recall window, in the $n$-person
game only $a_{t}$ is monitored (while the forecast $c_{t}$ is computed by
each player separately, but is \emph{not} played). Therefore, in order to
run the calibrated procedure, in the $n$-person game each player needs to
remember at time $T$, in addition to the last $R$ action combinations $%
a_{T-R},...,a_{T-1},$ also the last $R$ forecasts $c_{T-R},...,c_{T-1}.$
\textquotedblleft Finite recall" of size $R$ in the calibration procedure
therefore becomes \textquotedblleft finite memory" of size $2R$ in the game
dynamic: the memory contains $R$ elements of $A$ and $R$ elements of%
\footnote{%
For a similar transition from finite recall to finite memory, see Theorem 7
in Hart and Mas-Colell (2006).} $C.$

Alternatively, to get finite recall rather than finite memory one may
introduce an artificial player, say, player $0,$ with action set $A^{0}:=X$
and constant payoff function $u^{0}\equiv 0,$ who plays at each period $t$
the forecast $c_{t},$ i.e., $a_{t}^{0}=c_{t};$ this way the forecasts become
part of the recall of all players.

\emph{(f) Forecasting joint play. }In (D1) one may use a procedure that
forecasts the joint play: the forecasts $c_{t}$ lie in $\Delta (A),$ rather
than in $\prod_{i}\Delta (A^{i}))$ (the dimension is then larger, $%
\prod_{i}m^{i}$ instead of $\sum_{i}m^{i}).$ The approximate best reply
functions $g^{i}$ can then be defined over $\Delta (A),$ and the proof
carries through essentially without change. Thus most of the time the play
is close to Nash equilibrium, despite the fact that the forecasts are
allowed to be correlated; in fact, the forecasts turn out to be close to
being independent (because $g(c_{t})\in X$ is independent, and $c_{t}$ is
close to $g(c_{t})).$

\emph{(g) Separate forecasts.} One cannot simplify (D1) by replacing the
forecasting procedure that yields $c_{t}=(c_{t}^{i})_{i\in N}\in X$ with $n$
separate forecasting procedures that yield $c_{t}^{i}\in X^{i}$ for each $%
i\in N,$ because then behaviors $x_{t}$ and forecasts $c_{t}$ need no longer
be close (in part (iii) of the proof, when $c_{s}^{i}$ is now close to $%
c_{t}^{i}$ for some $i,$ it does not follow that $c_{s}^{j}$ and $c_{t}^{j}$
for $j\neq i$ are also close, and so neither are $g(c_{s})$ and $g(c_{t})$).

\emph{(h) Continuous approximate best reply}. In (D2) one may take the
functions $g^{i}$ to be continuous rather than Lipschitz and carry the proof
with the modulus of continuity instead of the Lipschitz bound (for
uncoupledness one would need to require uniform equicontinuity).

\emph{(i) Exhaustive search. }Dynamics that perform exhaustive search can
also be used to get the result of Theorem\footnote{%
We thank Yakov Babichenko for suggesting this.} \ref{th:finite}. Take for
instance a finite grid on $X$, say, $D=\{d_{1},...,d_{M}\}\subset X,$ that
is fine enough so that there always is a pure Nash $\varepsilon $%
-equilibrium on the grid. Let the dynamic go over the points $%
d_{1},d_{2},... $ in sequence until the first time that $d_{T}^{i}\in 
\mathrm{BR}_{\varepsilon }^{i}(d_{T}^{-i})$ for all $i,$ following which $%
d_{T}$ is played forever. This is implemented by having for every player $i$
a distinct action $a_{0}^{i}\in A^{i}$ that is played at time $t$ only when $%
d_{t}^{i}\in \mathrm{BR}_{\varepsilon }^{i}(d_{t}^{-i})$ (otherwise a
different action is played); once the action combination $%
a_{0}=(a_{0}^{i})_{i\in N}\in A$ is played, say, at time $T,$ each player $i$
plays $d_{T}^{i}$ at all $t>T.$ This dynamic is uncoupled (each player only
considers $\mathrm{BR}_{\varepsilon }^{i})$ and has memory of size $2$
(i.e., $2$ elements of $X)$: for $t\leq T,$ it consists of $d_{t-1}$ and $%
a_{t-1}$ (the last checked point and the last played action combination);
for $t>T,$ it consists of $d_{T}$ and $a_{0}$. Of course, all players need
to coordinate before playing the game on the sequence $d_{1},d_{2},...,d_{M}$
and the action combination $a_{0}.$

\emph{(j) Continuous action spaces. }The result of Theorem \ref{th:finite}
easily extends to continuous action spaces and approximate \emph{pure} Nash
equilibria. Assume that for each player $i\in N$ the set of actions $A^{i}$
is a convex compact subset of some Euclidean space (such games arise, for
instance, from exchange economies where the actions are net trades; see,
e.g., Hart and Mas-Colell 2015). Thus $A=\prod {}_{i\in N}A^{i}$ is a
compact convex set in some Euclidean space, say, $\mathbb{R}^{m}.$

For every $\varepsilon \geq 0,$ the set of $\emph{pure}$ $\varepsilon $\emph{%
-best replies} of player $i$ to $a^{-i}\in a^{-i}$ is $\mathrm{PBR}%
_{\varepsilon }^{i}(a^{-i}):=\{a^{i}\in A^{i}:u^{i}(a^{i},a^{-i})\geq
\max_{b^{i}\in A^{i}}u^{i}(b^{i},a^{-i})-\varepsilon \}.$ An action
combination $a\in A$ is a \emph{pure} \emph{Nash }$\varepsilon $\emph{%
-equilibrium} if $a^{i}\in \mathrm{PBR}_{\varepsilon }^{i}(a^{-i})$ for
every $i\in N;$ let \textrm{PNE}$(\varepsilon )\subseteq A$ denote the set
of pure Nash $\varepsilon $-equilibria.

Smooth calibrated learning is defined as above, except that now the
approximate best replies are pure actions (the play is $a_{t}=g(c_{t}),$ and
it is monitored by all players). Our result here is:

\begin{theorem}
\label{th:cont}Fix the finite set of players $N,$ the convex compact action
spaces $A^{i}$ for all $i\in N,$ and the Lipschitz bound $L<\infty .$ For
every $\varepsilon >0,$ and any smooth calibrated learning dynamic with
appropriate parameters, there is $T_{0}\equiv T_{0}(\varepsilon ,L)$ such
that for every $T\geq T_{0},$ 
\[
\frac{1}{T}\left\vert \{t\leq T:a_{t}\in \mathrm{PNE}(\varepsilon
)\}\right\vert \geq 1-\varepsilon 
\]%
for every game with payoff functions $(u^{i})_{i\in N}$ that are $L$%
-Lipschitz (i.e., $\mathcal{L}(u^{i})\leq L)$ and quasi-concave in one's own
action (i.e., $u^{i}(a^{i},a^{-i})$ is quasi-concave in $a^{i}\in A^{i}$ for
every $a^{-i}\in A^{-i}),$ for all $i\in N.$
\end{theorem}

\begin{proof}
We now have $A^{i}=X^{i}$ and $a_{t}=x_{t}=g(c_{t}),$ and everything is
deterministic. Proceed as in the proof of Theorem \ref{th:finite}, skipping
part (ii) (the use of the Law of Large Numbers) and taking $\varepsilon
_{1}=0$ and $T_{1}=0.$
\end{proof}

\bigskip

\emph{(k) Reaction function and fixed points. }The proof of Theorem \ref%
{th:finite} shows that in the leaky calibration game, if the A-player uses a
stationary strategy given by a Lipschitz \textquotedblleft reaction"
function $g$ (i.e., he plays $g(c_{t})$ at time $t)$, then smooth
calibration implies that the forecasts $c_{t}$ are close to fixed points of $%
g$ most of the time.

\appendix{}

\section{Appendix\label{s:appendix}}

Let $C$ be a compact subset of $\mathbb{R}^{m}$ and let $\varepsilon >0.$ A 
\emph{maximal }$2\varepsilon $\emph{-net} in $C$ is a maximal collection of
points $z_{1},...,z_{K}\in C$ such that $\left\Vert z_{k}-z_{j}\right\Vert
\geq 2\varepsilon $ for all $k\neq j;$ maximality implies $\cup
_{k=1}^{K}B(z_{k},2\varepsilon )\supseteq C.$ Let $\alpha
_{k}(x):=[3\varepsilon -\left\Vert x-z_{k}\right\Vert ]_{+},$ and put $\bar{%
\alpha}(x):=\sum_{k=1}^{K}\alpha _{k}(x).$ For every $x\in C$ we have $0\leq
\alpha _{k}(x)\leq 3\varepsilon $ and $\bar{\alpha}(x)\geq \varepsilon $
(since $\alpha _{k}(x)\geq \varepsilon $ when $x\in B(z_{k},2\varepsilon ),$
and the union of these balls covers $C).$ Finally, define $\beta
_{k}(x):=\alpha _{k}(x)/\bar{\alpha}(x).$

\begin{lemma}
\label{l:partition-of -unity}The functions $(\beta _{k})_{1\leq k\leq K}$
satisfy the following properties:

\begin{description}
\item[(i)] $\beta _{k}(x)\geq 0$ for all $x\in C\ $and all $k.$

\item[(ii)] $\sum_{k=1}^{K}\beta _{k}(x)=1$ for all $x\in C.$

\item[(iii)] $\beta _{k}(x)=0$ for all $x\notin B(z_{k},3\varepsilon ).$

\item[(iv)] For each $x\in C$ there are at most\footnote{%
We have not tried to get the best bounds in (iv) and (v); indeed, they may
be easily reduced.} $4^{m}$ indices $k$ such that $\beta _{k}(x)>0.$

\item[(v)] $\mathcal{L}(\beta _{k})\leq 4^{m+2}/\varepsilon $ for every $k.$
\end{description}
\end{lemma}

\begin{proof}
(i) and (ii) are immediate. For (iii), we have $\beta _{k}(x)>0$ iff $\alpha
_{k}(x)>0$ iff $\left\Vert x-z_{k}\right\Vert <3\varepsilon .$ This implies
that $B(z_{k},\varepsilon )\subseteq B(x,4\varepsilon ).$ The open balls of
radius $\varepsilon $ with centers at $z_{k}$ are disjoint (because $%
\left\Vert z_{k}-z_{j}\right\Vert \geq 2\varepsilon $ for $k\neq j),$ and so
there can be at most $4^{m}$ such balls included in $B(x,4\varepsilon )$
whose volume is $4^{m}$ times larger; this proves (iv). For every $x,y\in C$:%
\begin{eqnarray*}
\left\vert \beta _{k}(x)-\beta _{k}(y)\right\vert &\leq &\left\vert \frac{%
\alpha _{k}(x)}{\bar{\alpha}(x)}-\frac{\alpha _{k}(y)}{\bar{\alpha}(x)}%
\right\vert +\left\vert \frac{\alpha _{k}(y)}{\bar{\alpha}(x)}-\frac{\alpha
_{k}(y)}{\bar{\alpha}(y)}\right\vert \\
&\leq &\frac{1}{\bar{\alpha}(x)}\left\vert \alpha _{k}(x)-\alpha
_{k}(y)\right\vert +\frac{\alpha _{k}(y)}{\bar{\alpha}(x)\bar{\alpha}(y)}%
\sum_{j=1}^{K}\left\vert \alpha _{j}(x)-\alpha _{j}(y)\right\vert \\
&\leq &\frac{1}{\varepsilon }\left\Vert x-y\right\Vert +\frac{3\varepsilon }{%
\varepsilon \cdot \varepsilon }2\cdot 4^{m}\left\Vert x-y\right\Vert \leq 
\frac{4^{m+2}}{\varepsilon }\left\Vert x-y\right\Vert
\end{eqnarray*}%
(since $\bar{\alpha}(x)\geq \varepsilon $, $\alpha _{k}(x)\leq 3\varepsilon
, $ and there are at most $2\cdot 4^{m}$ indices $j$ where neither $\alpha
_{j}(x)$ nor $\alpha _{j}(y)$ vanishes); this proves (v).
\end{proof}

\bigskip

Thus, the functions $(\beta _{k})_{1\leq k\leq K}$ constitute a \emph{%
Lipschitz partition of unity} that is subordinate to the maximal $%
2\varepsilon $-net $z_{1},...,z_{K}$. Next, we obtain a basis for the
Lipschitz functions on $C.$

\begin{lemma}
\label{l:lipschitz-basis}Let $W_{L}$ be the set of functions $w:C\rightarrow
\lbrack 0,1]$ with $\mathcal{L}(w)\leq L.$ Then for every $\varepsilon >0$
there exist $d$ functions $f_{1},...,f_{d}\in W_{L}$ such that for every $%
w\in W_{L}$ there is a vector $\varpi \equiv \varpi _{w}\in \lbrack 0,1]^{d}$
satisfying 
\[
\max_{x\in C}\left\vert w(x)-\sum_{i=1}^{d}\varpi _{i}f_{i}(x)\right\vert
<\varepsilon . 
\]%
Moreover, one can take $d=\mathrm{O}(L^{m}/\varepsilon ^{m+1}).$
\end{lemma}

\begin{proof}
Put $\varepsilon _{1}:=\varepsilon /(3L).$ Let $z_{1},...,z_{K}$ be a
maximal $2\varepsilon _{1}$-net on $C,$ and let $\beta _{1},...,\beta _{K}$
be the corresponding Lipschitz partition of unity given by Lemma \ref%
{l:partition-of -unity} (for $\varepsilon _{1}).$

Given $w\in W_{L},$ let $v(x):=\sum_{k=1}^{N}w(z_{k})\beta _{k}(x);$ then $%
w(z_{k})\in \lbrack 0,1]$ and we have%
\begin{eqnarray*}
\left\vert w(x)-v(x)\right\vert &=&\left\vert \sum_{k=1}^{N}\left(
w(x)-w(z_{k})\right) \beta _{k}(x)\right\vert \leq \sum_{k:\beta
_{k}(x)>0}\beta _{k}(x)\left\vert w(x)-w(z_{k})\right\vert \\
&\leq &\sum_{k:\beta _{k}(x)>0}\beta _{k}(x)3\varepsilon _{1}L=3\varepsilon
_{1}L,
\end{eqnarray*}%
since $\beta _{k}(x)>0$ implies $\left\Vert x-z_{k}\right\Vert <3\varepsilon
_{1}$ and thus $|w(x)-w(z_{k})|\leq L\left\Vert x-z_{k}\right\Vert \leq
L\cdot 3\varepsilon _{1}$ (because $\mathcal{L(}w)\leq L).$

Now $\mathcal{L}(\beta _{k})\leq 4^{m+2}/\varepsilon _{1}$ by (v) of Lemma %
\ref{l:partition-of -unity}; we thus replace each $\beta _{k}$ by the sum of 
$Q=\left\lceil 4^{m+2}/(\varepsilon _{1}L)\right\rceil $ identical copies of 
$(1/Q)\beta _{k}$---denote them $f_{k,1},...,f_{k,Q}$---which thus satisfy $%
\mathcal{L}(f_{k,q})=(1/Q)\mathcal{L}(\beta _{k})\leq L,$ and so%
\[
\left\vert w(x)-\sum_{k=1}^{K}\sum_{q=1}^{Q}w(z_{k})f_{k,q}(x)\right\vert
=\left\vert w(x)-v(x)\right\vert \leq 3\varepsilon _{1}L=\varepsilon . 
\]%
The $d=KQ$ functions $(f_{k,q})_{1\leq k\leq K,1\leq q\leq Q}$ yield our
result.

Finally, $K=\mathrm{O}(\varepsilon _{1}^{-m})$ (because $C$ contains the $K$
disjoint open balls of radius $\varepsilon _{1}$ centered at the $z_{k})$
and $Q\leq 4^{m+2}/(\varepsilon _{1}L)+1,$ and so $d=KQ=\mathrm{O}%
(\varepsilon _{1}^{-m-1}L^{-1})=\mathrm{O}(\varepsilon ^{-m-1}L^{m})$.
\end{proof}

\bigskip

In the game setup we construct $\varepsilon $-best reply functions that are
Lipschitz. The following lemma applies when the action spaces are finite (as
in Theorem \ref{th:finite}), and also when they are continuous (as in
Theorem \ref{th:cont}). In the former $C=X=\prod_{i\in N}X^{i}$ where $%
X^{i}=\Delta (A^{i}),$ and in the latter $C=X=A=\prod_{i\in N}A^{i},$ and
the set $\Delta (A^{i})$ is identified with $A^{i}$; also, $\mathrm{BR}%
_{\varepsilon }^{i}$ stands for $\mathrm{PBR}_{\varepsilon }^{i},$ the set
of \emph{pure} $\varepsilon $-best replies.

\begin{lemma}
\label{p:BR}Assume that for each player $i\in N$ the function $%
u^{i}:X\rightarrow \mathbb{R}$ is a Lipschitz function with $\mathcal{L}%
(u^{i})\leq L,$ and $u^{i}(\cdot ,c^{-i})$ is quasi-concave on $X^{i}$ for
every fixed $c^{-i}\in X^{-i}.$ Then for every $\varepsilon >0$ there is a
Lipschitz function $g^{i}:X\rightarrow X^{i}$ such that $g^{i}(c)\in \mathrm{%
BR}_{\varepsilon }^{i}(c^{-i})$ for all $c\in X,$ and $\mathcal{L}%
(g^{i})\leq \nu _{m}(L/\varepsilon )^{m+1}$ where the constant $\nu _{m}$
depends only on the dimension $m.$
\end{lemma}

\begin{proof}
Put $\varepsilon _{1}:=\varepsilon /(6L).$ Let $z_{1},...,z_{K}\in C$ be a
maximal $2\varepsilon _{1}$-net on $C$, and let $\beta _{1},...,\beta _{K}$
be the subordinated Lipschitz partition of unity given by Lemma \ref%
{l:partition-of -unity}. For each $i\in N$ and $1\leq k\leq K$ take $%
x_{k}^{i}\in \mathrm{BR}_{0}^{i}(z_{k}^{-i}),$ and define $%
g^{i}(c):=\sum_{k=1}^{K}\beta _{k}(c)x_{k}^{i}.$ Because $\beta _{k}(c)>0$
if and only if $\left\Vert c-z_{k}\right\Vert <3\varepsilon _{1},$ it
follows that $x_{k}^{i}\in \mathrm{BR}_{\varepsilon }^{i}(c^{-i})$ (indeed,
for every $y^{i}\in \Delta (A^{i})$ we have $%
u^{i}(x_{k}^{i},c^{-i})>u^{i}(x_{k}^{i},z_{k}^{-i})-3L\varepsilon _{1}\geq
u^{i}(y^{i},z_{k}^{-i})-3L\varepsilon _{1}>u^{i}(y^{i},c^{-i})-6L\varepsilon
_{1}=\varepsilon ,$ where we have used $\mathcal{L}(u^{i})\leq L$ twice, and 
$x_{k}^{i}\in \mathrm{BR}_{0}^{i}(z_{k}^{-i})).$ The set $\mathrm{BR}%
_{\varepsilon }^{i}(c^{-i})$ is convex by the quasi-concavity assumption,
and so $g^{i}(c),$ as an average of such $x_{k}^{i},$ belongs to $\mathrm{BR}%
_{\varepsilon }^{i}(c^{-i}).$

Now $\max_{c\in C}||c||\leq \sqrt{m}$ (because $C\subseteq \lbrack
0,1]^{m}), $ and so $\left\Vert x_{k}\right\Vert \leq \sqrt{m}$ (where $%
x_{k}=(x_{k}^{i})_{i\in N})$ for all $k,$ and $K\leq (\sqrt{m}/\varepsilon
_{1})^{m}$ (because $C\subseteq B(0,\sqrt{m})$ contains the $K$ disjoint
open balls of radius $\varepsilon _{1}$ centered at the points $z_{k}).$
Therefore the Lipschitz constant of $g(c)=\sum_{k=1}^{K}\beta _{k}(c)x_{k}$
satisfies, by Lemma \ref{l:partition-of -unity} (v), $\mathcal{L}(g)\leq
\sum_{k=1}^{K}\left\Vert x_{k}\right\Vert \mathcal{L(}\beta _{k})\leq (\sqrt{%
m}/\varepsilon _{1})^{m}\,\sqrt{m}\,4^{m+2}/\varepsilon _{1}=\nu
_{m}\varepsilon ^{-m-1}L^{m+1}$ for $\nu _{m}=\sqrt{m}^{m+1}4^{m+2}6^{m+1}$.
\end{proof}

\end{document}